\documentclass{article}
   \usepackage[dvips]{graphicx}
   \usepackage{comment}
   \usepackage{pgfpages}
   \usepackage[T1]{fontenc}
   \usepackage[LY1]{fontenc}
   \usepackage[LY1]{fontenc}				
\usepackage{pdfsync}
\usepackage{hyperref}
   \usepackage{psfrag}
   \usepackage{amssymb,amsmath,amsthm, graphicx}
   \usepackage[LY1]{fontenc}				  
\usepackage{times}
\usepackage{amscd, amsfonts, tikz}

\newtheorem{theorem}{Theorem}[section]
\newtheorem{lemma}[theorem]{Lemma}

\newtheorem{proposition}[theorem]{Proposition}

\def\beq{\begin{equation}}
\def\eeq{\end{equation}}

\def\bH{\mathbf{H}}
\def\bI{\mathbf{1}}

\def\cA{\mathcal{A}}
\def\cB{\mathcal{B}}

\def\cG{\mathcal{G}}
\def\cH{\mathcal{H}}
\def\cI{\mathcal{I}}

\def\cK{\mathcal{K}}

\def\cM{\mathcal{M}}
\def\cN{\mathcal{N}}
\def\cO{\mathcal{O}}

\def\cR{\mathcal{R}}
\def\cS{\mathcal{S}}
\def\cT{\mathcal{T}}
\def\cV{\mathcal{V}}
\def\cW{\mathcal{W}}

\def\cZ{\mathcal{Z}}

\def\tf{\tilde{f}}

\def\tZ{\tilde{Z}}

\def\tmu{\tilde\mu}

\newcommand{\eps}{\varepsilon}

\def\sab{S_{\alpha,\sigma}}
\def \E{\mathbf{E}}
\def \P{\mathbf{P}}
\def \Pois{\text{Poiss}}
\def\({\left(}
\def\rt){\right)}
\def\[{\left[}
\def\]{\right]}
\def\dsim{\stackrel{d}{\sim}}

\def \R{\mathbb{R}}

\def \N{\mathbb{N}}

\def \P{\mathbf{P}}
\def \E{\mathbf{E}}
\def \aa{\alpha}

\def \eps{\epsilon}

\def \ff{\infty}

\def \({\left(}
\def \){\right)}

\def \beq{\begin{equation}}
\def \ee{\end{equation}}
\def \bea{\begin{eqnarray}}
\def \eea{\end{eqnarray}}
\def \bes{\begin{eqnarray*}}
\def \ees{\end{eqnarray*}}

\def \nn{\nonumber}
\def \aaa{\mathfrak	{a}}
\def\tP{\tilde{P}}
\def\tZ{\tilde{Z}}
\numberwithin{equation}{section}

\begin{document}

\title{Stable laws for chaotic billiards with cusps at flat points
\thanks{ {\bf Received date}: }
}
\author{PAUL JUNG\,\,\,\,AND\,\,\,\,  HONG-KUN ZHANG}

\date{}


\maketitle
\begin{abstract}  We consider billiards with a single cusp where the walls meeting at the vertex of the cusp have zero one-sided curvature, thus forming a flat point at the vertex. For H\"older continuous observables, we show that properly normalized Birkhoff sums, with respect to the billiard map, converge in law to a totally skewed $\alpha$-stable law.\end{abstract}
\tableofcontents

\section{Introduction}

Here we consider  dispersing billiards with  a cusp at a flat point, similar to that which was discussed in \cite{Z2016b}.  More precisely, for any fixed constant $\beta>2$ (which will determine the sharpness of the cusp), we consider a dispersing billiard table $Q=Q_{\beta}$  with a boundary consisting of a finite number of $C^3$ smooth curves $\Gamma_i$, $i=1\cdots, n_0$, where $n_0\geq 3$, with a cusp formed by two of those curves and such that there is a `perpendicular opposing point' to the cusp (see below). To simplify things, we have assumed the table has a single cusp at $P=\Gamma_1\cap\Gamma_2$;  both $\Gamma_1$ and $\Gamma_2$ have zero  derivatives up to $\beta-1$ order  at $P$, and the $\beta$-order derivative is not zero; we also assume that all other boundary components are dispersing and have curvature bounded away from zero.

 We  choose a Cartesian coordinate system $(s,z)$ originated at $P$, with the horizontal $s$-axis being the tangent line to both $\Gamma_1$ and $\Gamma_2$. Assume $\Gamma_1$ and $\Gamma_2$ can be represented as
 \beq\label{z1s}
 z_1(s)=\beta^{-1} s^{\beta},\,\,\,\,\,\,{z_2(s)=-\beta^{-1} s^{\beta}}\eeq
 for $s\in [0,\eps_0]$ with $\eps_0>0$ being a small fixed number.

We investigate a limit law for the billiard system on $Q_{\beta}$. To simplify our analysis, we denote $\Gamma_3$ as the opposition side to the cusp. Following a similar assumption in \cite{Z2016b}, we also  suppose that the (unique) horizontal trajectory running out of the cusp from $P$ will hit the boundary $\Gamma_3$  perpendicularly, i.e., at a perpendicular opposing point. Let us note that when extending the analysis below to multiple cusps, it is important for each cusp to have a perpendicular opposing point.

The billiard flow $(\Phi^t)$ is defined on the unit sphere bundle $Q\times\mathbf{S}^1$ and preserves
Liouville measure. There is a natural cross section $\cM$ in	 $Q\times	 \mathbf{S}^1$ that contains all post-collision vectors based at the boundary of the table $\partial Q$. The set $\cM=\partial Q \times [0,\pi]$ is called the collision space.
The corresponding  billiard map $  T : \cM\to \cM$ takes a vector $x\in \cM$ to the next post-collision vector	 along the trajectory of $x$.
Let the set $S_0$ consist of all grazing collision vectors with walls as well as all collision vectors at corner points. Then $\cS:=S_0\cup   T       ^{-1}S_0$ is called the {\it{singular}} set of $T$.
The billiard map $T: \cM\setminus	\cS\to \cM\setminus   T       \,\cS$ is a local $C^{2}$ diffeomorphism and preserves a natural
absolutely continuous probability measure $d\mu=\frac{1}{2|\partial Q|} \sin\varphi\, dr\, d\varphi$ on the collision space $\cM=\{(r,\varphi)\}$ (here $|\partial Q|$ is the length of $\partial Q$).

Any post-collision vector $x\in \cM$ can be represented by $x=(r, \varphi)$, where $r$ is the arclength parameter
along $\partial Q$,	 and $\varphi\in [0, \pi]$ is the angle formed by  the tangent line of the boundary and the collision vector in the clockwise direction. For  simplicity, we assume the cusp point has $r$-coordinate
\beq\label{r'}
r=r'\quad\text{and}\quad r=r'',\eeq with respect to $\Gamma_1$ and $\Gamma_2$, respectively (we will also later set $r'=0$).
We define the subset $M\subset \cM,$ which consists of all collisions on $\partial Q\setminus(\Gamma_1\cup\Gamma_2)$. We then define $F: M\to M$ (see \eqref{Fdef} below) as the first return map, such that for any $x\in M$, $Fx\in M$ is the first return to $M$ along the forward iterations of $T $. It is known that $F$ preserves the measure
\beq\label{hmu}
\tmu:=\frac{1}{\mu(M)}\mu|_M.
\eeq

Rigorous bounds on the decay of correlations for billiards with flat points were derived recently
in \cite{Z2016b}, where a detailed description of
billiards with flat points is also given. It was shown that if $f, g$ are  H\"older continuous functions
 on the collision space $\cM$, then for all $n\in \mathbb{Z}$,
\beq\label{cCn}\mu(f\circ   T^n\cdot g)-\mu(f)\mu(g)= \cO(n^{\frac{1}{1-\beta}})\eeq Here we use the standard notation $\mu(f)=\int_{\cM} f\,d\mu$. It is the above slow decay of correlations that leads one to expect limiting behavior, which is different from the classical Central Limit Theorem,
in the Birkhoff sums
$$ \cS_n f := f + f\circ   T   +\cdots +f\circ   T  ^{n-1}$$
for H\"older continuous functions on $\cM$.
As usual, we consider the case {$\mu(f) = 0$}, where the general case follows by simply subtracting off $\mu(f)$.

Letting $\alpha=\frac{\beta}{\beta-1}$ (equivalently, $\beta=\frac{\alpha}{\alpha-1}$), one can check that $\alpha \in (1,2)$. Our main goal is to establish an $\alpha$-Stable Limit Theorem for the sequence $\{\cS_n f, n\ge 0\}$. Indeed, a {function $f$} with $\mu(f)=0$ is said to be in the domain of attraction of a (strictly) $\alpha$-stable law if there exists $\{b_n\}$ such that $\{\frac{\cS_nf}{b_n}\}$ converges in distribution to a random variable $S_\alpha$ with an $\alpha$-stable law. Here, {\em strictly} simply means that $\mu(f)=0$, and we shall henceforth just say $\alpha$-stable.  In particular, then there exist constants {$C, C^-\ge 0$} such that the limiting stable law satisfies
\beq
\lim_{x\to\infty}x^\alpha\P(S_\alpha>x)=C, \quad \lim_{x\to\infty}x^\alpha\P(S_\alpha<-x)=C^-.
\eeq
{If $f$ is positive in the neighborhoods of $r'$ and $r''$ corresponding to the vertex of the cusp $P$, then $C>0$ and $C^-=0$ (this holds under the weaker condition $I_f>0$ defined in \eqref{If}), and the stable law is said to be {\it totally skewed}. We will consider only totally skewed stable laws in what follows (without loss of generality, positively skewed so that $C>0, C^-=0$).}

The constant $C$ above  determines a so-called {\em scale parameter} $\sigma>0$ which plays
a role analogous to the standard deviation of a Gaussian distribution. In particular  (see \cite[p.17]{samorodnitsky1994stable}), $$\sigma^\alpha={C\,\Gamma(2-\alpha)\cos(\pi\alpha/2)}{(1-\alpha)^{-1}}.$$
We will henceforth denote by $\sab$, a stable random variable with characteristic function
\beq
\E\(e^{i u \sab}\right)=\exp\(-|u\sigma|^\alpha\(1-i \text{sign}(u)\tan\frac{\pi\alpha}{2}\right)\right).
\eeq

{For any $\gamma\in (0,1)$, we denote $\cH_{\gamma}$ as the class of all H\"older continuous functions $f:\cM\to\R$, with H\"older exponent $\gamma$.}


\begin{theorem}[Stable Limit Theorem for billiards with a cusp]\label{thm:1}
Let $Q_\beta$, where $\beta\in(2,\infty)$, be a  billiard table with cusp defined by (\ref{z1s}) and suppose $f\in \cH_{\gamma}$ for some $\gamma>0$. Suppose $\mu(f)=0$ and
\beq\label{If}
I_f:=\frac{1}{4}\int_{0}^{\pi} ( f(r',\varphi)+f(r'',\varphi))\sin^{\frac{1}{\alpha}}\varphi\, d\varphi\neq 0\eeq
where  $r',r''$ are as in \eqref{r'}.
Then as $n\to\infty$,
\begin{equation}\label{stablelaw}
\frac{\cS_{n} f}{n^{1/\alpha}}\xrightarrow{d} S_{\alpha, \sigma_f}
\end{equation}
with $\alpha=\frac{\beta}{\beta-1}$ and $\sigma_f^{\alpha}=\frac{2 I_f^{\alpha}}{{\beta|\partial Q|}}$.\end{theorem}

\noindent{\bf Remark:}
 The result extends easily to bounded, piecewise H\"older functions which are H\"older continuous in  a neighborhood of the region in $\cM$ corresponding to the cusp (at $r'$ and $r''$) and whose discontinuities are contained in the singular set of $T$.

A word on the method of proof of the above theorem. We follow a recently popular approach in studying the statistical properties of $(T  ,\cM)$, namely we use an inducing scheme as introduced in \cite{markarian2004billiards, chernov2005billiards}. By removing spots with weak hyperbolicity from the phase space, one considers first the return map on the subspace $M\subset \cM$.
More precisely, we define $M$ to be the collision space on dispersing boundaries $\Gamma_i$, $i=3,\cdots, n_0$,
as well as those collision vectors on $\Gamma_i$, $i=1,2$, such that the number of forward collisions on the two boundary components $i=1,2$, before hitting a boundary component with $i\ge 3$, does not exceed $K_0$, for some fixed $K_0\geq 1$. For any $x\in M$
we call $$\cR(x) := \min\{n\geq	 1:   T^n(x)\in M\}$$ the {\em first return time function} and the  {\em return map} $F\colon M \to M$ is defined by
\beq \label{Fdef}
     F(x) :=   T  ^{\cR(x)}(x),\,\,\,\,\,\text{for all } x\in M.
\eeq
The return map $F$ preserves the conditional measure $\tmu$, defined earlier. It was proved in \cite{Z2016b} that the induced system $(F,M,\tmu)$ is uniformly hyperbolic and enjoys exponential decay of correlations -- a property which aids tremendously in proving probabilistic limit theorems.

The basic outline of the proof is that we first prove the Stable Limit Theorem
on $M$  first for the special case of a centered version of the return time function using exponential decay of correlations on $M$ -- this special case introduces the basic ideas involved in the proof of the main theorem.
 Next we show that a general {\em induced} function
\beq\label{inducef}
  \tf(x) :=\sum_{k=0}^{\cR(x)}
f(T^k x), \quad x\in M,\eeq
can always be approximated by a return time function, and consequently extend the result for return time functions to general induced functions.  Finally, using the idea of \cite{balint2011limit}, we lift the limit theorem from the induced system to the original system.
We do this in our context, by an application of the Continuous Mapping Theorem.

The rest of the paper is organized as follows. In Section \ref{sec:prelim} we gather and review some preliminary tools concerning billiards with cusps and H\"older bounds for induced functions. In Section \ref{sec:return}, we prove convergence to a stable random variable for the partial sums of iterations of the return time function, which should be considered a special case of an induced function.
In Section \ref{sec:induced} we extend this result to general induced functions on $M$. In Section \ref{sec:main thm} we ultimately extend this to convergence to $\alpha$-stable laws for functions on the original space $\cM$. In the last section, we provide technical calculations for the scale parameter $\sigma$. Throughout the proofs $C, C', C'', C_1, C_2, \ldots$ denote positive constants that may change in different paragraphs, and sometimes within the same proof.

\section{Preliminaries}\label{sec:prelim}
\subsection{Properties of billiards with cusps}

Billiards with cusps were previously studied in \cite{chernov2007dispersing} and \cite{Z2016b}, and the key analysis was a careful consideration of trajectories in the cusp which were called {\it corner series}. These trajectories play an important role here as well, and our first lemma below summarizes the key estimates concerning such trajectories, including an estimate which nearly gives us power law tails for the return times -- a key ingredient in the proofs below.

Let $(F,M,\tmu)$ be the induced system, with singularity set $\tilde{S}$.
 It easily follows from Kac's formula that $\tmu(\cR)= 1/\mu(M)$ which shows that $\cR$ has finite mean.  Our first important property
states among other things, that furthermore, the return time function $\cR$
satisfies the polynomial tail bound \beq \label{Mpol}
	\tmu(x\in M:
	 \cR(x)>n)\sim
	\,  n^{-\alpha},
\eeq where $\alpha=1+\frac{1}{\beta-1}$ is the constant of Theorem~\ref{thm:1} ($a_n\sim b_n$ denotes $C_1 b_n<a_n<C_2 b_n$ for some  constants $0<C_1<C_2<\infty$).     We will need even more precise estimates than this.

Let the domains $(M_N, N\ge 1)$, called $N$-cells, be defined by
\beq\label{Ncells}
M_N :=\{x\in M \,: \, \cR(x) = N+1\},\quad N\in\mathbb{N},
\eeq
In words, they consist of points $x$ whose forward trajectory enters the cusp and exits after exactly $N$-iterations inside the cusp.
It will be convenient to deal only with angles in $[0,\frac{\pi}{2}]$, so we denote
\beq\label{eta}
\eta_n:=\min(\varphi_n, \pi-\varphi_n)
\eeq
 where $ (r_n,\varphi_n)\equiv T^n x$, for $n=1,\cdots, N$.

Let $T^n x=(r_n,\varphi_n)$ be sufficiently close to the vertex $P$ of the cusp, and in the $(s,z)$-coordinate system described above \eqref{z1s}, denote $s_n$ as the $s$-coordinate of the point $T^n x$. If we set $r'=0$, one can check that
on the side of the cusp corresponding to $r'$,
\beq\label{sn}
s_n=r_n+\cO(r_n^2).
\eeq
On the side of the cusp corresponding to $r''$, similarly we have $s_n=r_n-r''+\cO(r_n^2)$, but by symmetry we will henceforth focus on the $r'$-side of the cusp.
We also have that {$(1, s_n^{\beta-1})$} is a tangent vector of $\partial Q$ at $s_n$, and
$$\rho_n:=\tan^{-1}(s_n^{\beta-1})$$ is the angle of the tangent vector at $s_n$ made with the horizontal axis, or equivalently, with the tangent line through the flat point $P$.

On $M_N$, we have $\cR=N$. Using notation similar to that of \cite{Z2016b}, define $\bar N$ by way of
$$\rho_{\bar N}:=\min\{\rho_n\,:\, 1\leq n\leq N\}.$$
It was argued in \cite{chernov2007dispersing} that on $M_N$, one has $|\bar N-N/2|\leq 2$ so that $\bar N$ is close to the middle of $N$.
We further subdivide the corner series into three segments. Fix a small value
$\bar\eta>0$ and let $$N_1=\max\{ n\leq \bar N\,:\, \eta_n<\bar\eta\},\,\,\,\, N_3=\min\{n\geq \bar N\,:\, \eta_n<\bar\eta\},$$ and also put $N_2=\bar N$.
In previous works, the segment on $[1, N_1]$ was called the ``entering period'' in the corner series, the
segment $[N_1 + 1;N_3-1]$ the ``turning period'', and the segment $[N_3,N]$
its ``exiting period''.

\begin{lemma}[Trajectory and strip estimates, \cite{Z2016b}]\label{Mm} For any $N\geq 1$, points in the cell $M_N$ have the following properties:  \\
{(1)  $N_1\sim N_2-N_1\sim N_3-N_2\sim N-N_3\sim N$, i.e. all the three segments in the corner series have length of order $N$;\\
(2) $\rho_1\sim N^{-\frac{\alpha}{\alpha+1}},\,\,\,\, \rho_{n}\sim n^{-1}\sim N^{-1}$, for $n\in [N_1, N_2]$;\\
(3) $\rho_n\sim n^{-\frac{1}{\alpha+1}}N^{-\frac{\alpha}{\alpha+1}}$, for $n\in [1,N_1]$;\\
(4) $\eta_1=\cO(N^{-\frac{\alpha}{\alpha+1}}), \,\,\eta_2\sim N^{-\frac{\alpha}{\alpha+1}}$;\\
(5) $\eta_n\sim (n N^{-1})^{\frac{\alpha}{\alpha+1}}$, for $n\in [1,N_1]$;\\
(6) the cell $M_N$ has width $\sim N^{-\frac{\alpha^2+\alpha+1}{\alpha+1}}$, length $\sim N^{-\frac{\alpha}{\alpha+1}}$, and density $\sim 1$; \\
(7) $\tilde\mu(M_N)\sim N^{-1-\alpha}$};\\
(8) The forward iterations of $T^kM_N$, $k=1,\cdots, N$, are all contained in a single strip $H_N$, bounded by two curves with equation $|r-r_f|^{\beta}\cos\varphi=C_N$, where $C_N\sim N^{-\alpha}$.
\end{lemma}

Note that part (7) above implies \eqref{Mpol}. In order to prove stable convergence, we will use Lemma \ref{Mm} above to prove, in Section \ref{sec:scale}, a refinement of \eqref{Mpol} which is given in the next lemma.
\begin{lemma}[Power law return times]\label{Mm2}
The return time function satisfies
\beq\label{eq:scale parameter}\lim_{n\to\infty} n\tmu\left(x\in M:\cR(x) >n^{\frac{1}{\alpha}}\right)=\frac{2I_1^\alpha}{\beta\mu(M)|\partial Q|}\eeq
where
{$I_1=\int_{0}^{\pi/2}\sin^{\frac{1}{\alpha}}\varphi\,d\varphi.$}
\end{lemma}
{Note that if $I_f=I_1$} in Theorem \ref{thm:1}, then the limit above is exactly $\sigma^\alpha$.

\subsection{H\"{o}lder continuity of $\tf$}

We divide $\cM$ into horizontal homogenous strips as introduced by Sinai.  More precisely, one divides $\cM $ into countably many sections (called
\emph{homogeneity strips}) defined by
$$
	\bH_k:=\{(r,\varphi)\in \cM \colon \pi/2-k^{-2}<\varphi <\pi/2-(k+1)^{-2}\},
$$
and
$$
	\bH_{-k}:=\{(r,\varphi)\in \cM \colon -\pi/2+(k+1)^{-2}<\varphi < -\pi/2+k^{-2}\},
$$
for all $k\geq k_0$ and \beq \label{bbH0}
	\bH_0:=\{(r,\varphi)\in \cM \colon -\pi/2+k_0^{-2}<\varphi <
	\pi/2-k_0^{-2}\}.
\eeq Here $k_0 \geq 1$ is a fixed (and usually large) constant.
For any $N\geq 1$, an  unstable curve $W\subset FM_N$ is called a homogeneous unstable curve for the induced map $F$,
if $T^{-k}W$ is contained by a single homogeneous strip for any given $k=0,1,\cdots, N$.
Next, let us discuss the regularity of  unstable curves for the induced map.
We fix a constant $C_b$, and denote $\cW^u_H$ as the collection of
all homogeneous unstable curves with curvature bounded by $C_b$
for the induced system $(F,M,\tmu)$.
We have the following result which says that curves in $\cW^u_H$  have uniform distortion bounds.

\begin{proposition}\label{distortionbound}
Let $W\in\cW^u_H$ be an unstable curve. Then there exists $C=C(Q)>0$ such that
\beq\label{distortion}|\ln J_WF^{-1}(x)-\ln J_WF^{-1}(y)|\leq C d_W(x,y)^{\gamma_0},\eeq for any $\gamma_0\in (0,1/3)$.
\end{proposition}
The proof of this proposition can be found in  Section \ref{sec:7}.

Next we study the H\"older continuity of the induced function $\tf$.  As before, for $\gamma\in (0,1)$, let $\cH_{\gamma}$	be the set of $\gamma$-H\"older functions. We denote the H\"older norm by
$$\|f\|_{\gamma}:= \sup_{ x, y}\frac{|f(x)-f(y)|}{d(x,y)^{\gamma}}<\infty,
$$
where $d(x,y)$ denotes distance.
For every $f\in \cH_{\gamma}$ we also define
\beq \label{defCgamma}
\|f\|_{C^{\gamma}}:=\|f\|_{\infty}+\|f\|_{\gamma}.
\eeq

For future reference, we record here a bound from \cite[Theorem 3]{CZ09}. For  \mbox{$f, g\in \cH_{\gamma}$,} and any integer $k$, the correlations of $f$ and $g\circ   F^k$ satisfy:
\beq\label{correT}\text{Cov}(f,g\circ   F^k):=|\tmu(f\cdot g\circ   F^k)-\tmu(f)\tmu(g)|\leq C \|f\|_{C^{\gamma}} \|g\|_{C^{\gamma}}\vartheta^k\eeq
where $C>0$ and $\vartheta\in (0,1)$ are constants.


We now estimate the H\"{o}lder norm of $\tf$, for any $f\in \cH_{\gamma}$.
\begin{lemma}\label{Holdertf}
For $\gamma\in (0,1)$, $f\in \cH_{\gamma}$, $N\geq 1$, and $x,y\in M_N$, the induced function $\tf$ has a H\"{o}lder norm given by the following:
\beq\label{holdertf}
|\tf(x)-\tf(y)|\leq C\|f\|_{\gamma}  N^{1+\gamma}d(x,y)^{\frac{\gamma}{\beta}}
\eeq
where $C=C(\gamma)>0$ is a constant.
\end{lemma}
\begin{proof} For any $N\geq 1$ and $x,y\in M_N$,
$$|\tf(x)-\tf(y)|\leq \sum_{k=0}^{N-1}\|f\|_{\gamma} d(T^k x, T^k y)^{\gamma}.$$
The images $\{T^k(M_N), k=1,\cdots, N-1\}$  stretch in the unstable
direction and shrink in the stable direction, as $k$ increases, thus we can assume that $x, y$ lie on one unstable curve $W\subset M_N$.
We first review a relation between the Euclidean norm and the $p$-norm in the tangent space:

\beq\label{penorm}\frac{\|dx_{1}\|}{\|dx\|}
=\frac{\|dx_{1}\|_p}{\|dx_{}\|_p}\cdot \frac{\cos\varphi_{}}{\cos\varphi_{1}}\cdot\sqrt{\frac{1+(d\varphi_{1}/dr_{1})^2}{1+(d\varphi_{}/dr_{})^2}}\eeq
where $x=(r,\varphi)$, $x_m=T^mx=(r_m,\varphi_m)$, for $m=1,\cdots, N$.

It was shown in \cite{Z2016b}~Lemma 17, that for $k=1,\ldots, N_1$, the unstable manifolds $\{T^kW\}$ at the points $\{T^k x\}$ expand under $T$ by a factor $1+\lambda_k$ in the $\|\cdot\|_p$-norm, with
\beq\label{lambda1}\lambda_k= \frac{\beta/(2\beta-1)}{k+C' \ln k+C''}+o(N^{-1}).\eeq
Similarly, by \cite{Z2016b}~Lemma 18, for $k=N_3,\ldots, N-1$, the unstable manifolds $\{T^kW\}$ at $\{T^k x\}$ are expanded under $T$ by a factor $1+\lambda_k$ in the $\|\cdot\|_p$-norm, with
\beq\label{lambda2}\lambda_k= \frac{(\beta-1)/(2\beta-1)}{k+C' \ln k+C''}+o(N^{-1}).\eeq
Finally, from \cite{Z2016b}~Proposition 16, it follows that for $k=N_1,\ldots, N_3$, the expansion factor $1+\lambda_k$ satisfies the asymptotic $\lambda_k \sim 1/k$.

On the other hand, it follows from \cite{Z2016b}~Proposition 4, that $\lambda_0$ and $\lambda_{N-1}$ can be arbitrarily large, since $x$ or $Fx$
may be asymptotically tangential to $\partial Q$.
For these two iterations with unbounded expansion factors, we instead
use the H\"{o}lder continuity of the original billiard
map $T$ near tangential collisions with the flat point:
\beq\label{holder bound}
d(T x, T y) \leq Cd(x, y)^{\frac{1}{\beta}}
\eeq
for some $C>0$ (see Eqn. (\ref{z1s})).

Let $\hat N$ be as in
 \eqref{eq:nhat2} in the proof of Proposition \ref{distortionbound}.
By using (\ref{lambda1}), (\ref{lambda2}), (\ref{penorm}) and the bounded distortion -- Proposition \ref{distortionbound},
we have for $m\in [1,\hat N]$
$$d(T^mx, T^m y)\leq C_0d(Tx,Ty)\prod_{l=1}^{m}(1+\lambda_l)\cdot\frac{\cos\varphi_l}{\cos\varphi_{l+1}}\leq C_1 m^{\frac{\beta-1}{2\beta-1}}d(Tx,Ty) $$
where $C_0>0$ depends on the distortion bound.

On the other hand, for $m\in [\hat N, N]$,
 \begin{align*}
 d(T^mx, T^m y)&\leq C_0d(Tx,Ty)\prod_{l=1}^{N}(1+\lambda_l)\cdot\frac{\cos\varphi_l}{\cos\varphi_{l+1}}\cdot\sqrt{\frac{1+(d\varphi_{l+1}/dr_{l+1})^2}{1+(d\varphi_{l}/dr_{l})^2}}\\
 &\leq C_1  N^{\frac{\beta-1}{2\beta-1}}N^{\frac{\beta}{2\beta-1}}d(Tx,Ty) \cos\varphi_1=C_1  N d(Tx,Ty) \cos\varphi_1,\end{align*} according to the estimation for expansion factor in \cite{Z2016b}~Proposition 2.

Combining the above facts, we have
\begin{align*}
|\tf(x)-\tf(y)|&\leq \sum_{k=0}^{N}\|f\|_{\gamma} d(T^kx,T^ky)^{\gamma}\\
&\leq \|f\|_{\gamma} C' d(Tx,Ty)^{\gamma}N^{1+\gamma}\\
&\leq \|f\|_{\gamma} C'' N^{1+\gamma}d(x,y)^{\frac{\gamma}{\beta}}.\end{align*}
This implies that $\tf$ has a H\"{o}lder norm of order $N^{1+\gamma}$, and has H\"{o}lder exponent $\frac{\gamma}{\beta}$.
\end{proof}

\subsection{Standard families}
We next review the concept of a standard pair and state a growth lemma. For an unstable curve $W$ and a probability measure $\nu_0$ on the Borel $\sigma$-algebra of $W$, we say that the pair $(W,\nu_0)$ is a \textit{standard pair} if $\nu_0$ is absolutely continuous with respect to the Lebesgue measure, $ m_W$, induced by the curve length, with density function $f(x):=d\nu_0/dm_W$ satisfying
\beq\label{lnholder}
|\ln f(x)-\ln f(y)|\leq C {d}_W(x,y)^{\gamma_0}.
\eeq
 Here  $\gamma_0$ is a fixed H\"{o}lder exponent which appears in the distortion bound, see Proposition \ref{distortionbound}. Also, ${d}_W(x,y)$ is the distance between $x$ and $y$ measured along the smooth curve $W$.

The notion of a standard pair was studied  by Chernov and Dolgopyat in \cite{CD}. In particular, they considered families of standard pairs
 $\cG=\{(W_\aaa, \nu_\aaa)\,:\, \aaa\in \cA\}$ where $\cA\subset [0,1]$.
Let $\cW=\{ W_{\aaa}\,|\, \aaa\in \cA\}$. We call $\cG$ a {\em standard family} if $\cW$ is a measurable foliation of a  measurable subset of $M$,
and there exists a finite Borel measure $\lambda_{\cG}$ on $\cA$, which defines a measure $\nu$
 on $M$  by
\beq\label{cGnm}
   \nu(B):=\int_{\aaa\in\cA} \nu_{\aaa}(B\cap
   W_{\aaa})\,
   d\lambda_{\cG}(\aaa)\hspace{1cm}
   \eeq
for all  measurable sets $B\subset M$. In the following, we denote a standard family by $\cG=(\cW,\nu)$.

Define a function $\cZ$ on standard families, such that
for any standard family $\cG=(\cW,\nu)$,
\beq\label{cZ}
\cZ(\cG):=\frac{1}{\nu(M)}\,\int_{\aaa\in\cA}|W_{\aaa}|^{-1}\,\lambda_{\cG}(d\aaa). \eeq

For any unstable curve $W\in \cW$, any $x\in W$,  and any $n\geq 1$, let $W^{k}(x)$ be the smooth unstable curve  in $F^k W$ that contains $F^k x$.
We define $r_k(x)$ as the minimal distance between $F^k x$ and  the two end points of $W(F^k x))$, measured along $W^k(x)$.
 According to the Growth Lemma in \cite[Lemma 9]{Z2016b}, we know that an $\eps$-neighborhood of the singular set of $F$ has measure of order $\eps^q$, with
\beq\label{q}
 q=\frac{\alpha(\alpha+1)}{\alpha^2+\alpha+1}.
 \eeq
\begin{lemma}[Growth Lemma]\label{growthlemma}
Let $\cG=(\cW, \nu)$ be a standard family such that \mbox{$\cZ(\cG)<\infty$.} Then for any $\eps>0$ and $k\ge 0$,
$$\nu(r_k(x)<\eps)\le C_0\eps^q \cZ(F^k\cG)\leq C_1(\vartheta^{k-1}\cZ(F\cG)+C_2)\eps^q$$
where $C_0>0, C_1>0, C_2>0$ and $\vartheta\in (0,1)$ are constants.
\end{lemma}

For a fixed large constant $C_\text{prop}>0$ (to be chosen in \eqref{cq} below), any standard family $\cG$ with $\cZ(\cG)<C_\text{prop}$ will be called a {\it{proper}} family. The following was proved in \cite[Theorem 2]{CZ09}. Denote  $F_* \nu$ as the push-forward measure.

 \begin{lemma}[Equidistribution]\label{equidistribution}
If $\cG=(\cW,\nu)$ is a proper family, then for any $g\in \cH_{\gamma}$ with $\gamma\in (0,1)$ and $k\ge 0$,
\beq\label{properg}
|F_*^k\nu(g)-\tilde\mu(g)|\leq C \|g\|_{C^{\gamma}}\vartheta^k.
\eeq
\end{lemma}

\section{Preliminary case: the return time function}\label{sec:return}

In this section we prove convergence in distribution to a stable random variable for the normalized partial sums of iterations of the centered return time function; this special case gives insight into the basic ideas behind the main theorem.
Indeed the return time is a special case of an induced function $\tf_0=\cR-\tmu(\cR)$ for
\beq\label{defnf0}f_0:=1-\frac1{\mu(M)}\bI_M\eeq
which satisfies $\mu(f_0)=0$.  One can check that, in general, $\mu(f)=0$ implies $\tmu(\tf)=0$. Note that although $f_0$ is not  H\"older continuous, it is H\"older continuous on a neighborhood of the cusp (at both $r'$ and $r"$), as well as piecewise H\"older, which is good enough for our purposes (see the remark following Theorem \ref{thm:1}).

Denote the Birkhoff sums of an induced function $\tf$, under the induced map $F$, by $S_n  \tf$:
$$
S_n  \tf :=  \tf+ \tf\circ   F+\cdots + \tf\circ   F^{n-1}.
$$

\begin{theorem}[Stable limits for the return time function]\label{thm:2}  Let $\cR$ be the first return time function on $M$. Then
\beq\label{stableR2}
\frac{S_n (\cR - \tmu(\cR))}{\sqrt[\alpha]{n}}\xrightarrow{d}  S_{\alpha,\tilde\sigma_\cR}
\eeq
where $\tilde\sigma_{\cR}^{\alpha}=\frac{2I_1^\alpha}{\beta\mu(M)|\partial Q| }$ and
{$I_1=\int_{0}^{\pi/2}\sin^{\frac{1}{\alpha}}\varphi\,d\varphi.$}
\end{theorem}
Note that {$I_1 =I_f$} in the special case where one takes $f=1-\frac{1}{\mu(M)}\bI_M$.


In the rest of the section we prove the above Theorem \ref{thm:2}, up to Lemma \ref{Mm2} which calculates the scale parameter  and which requires further technical calculations (in actuality, we  also refer to Lemma \ref{LHzero} in the proof of Theorem \ref{thm:2}, but a simplified version of this lemma for return times easily follows from the arguments in this section-- we omit the details).

From a probabilistic viewpoint, the two properties that lead to such convergence are (a) the power law tails of the return time function and (b) the  exponential decay of correlations for return time functions.
Once these two properties are established, one obtains a purely probabilistic Poisson-type convergence (Subsection \ref{sec:poisson}) which then leads to stable convergence.
 Lemma \ref{Mm2} establishes (a), while Lemma \ref{boundedcov} below will give (b).

Fix a finite union of open intervals $\cI=\cup (a_n,b_n)$.   
Our correlation bounds will depend on sets of the form $$A_{n,j}:=\{x\in M:\frac{1}{\sqrt[\alpha]{n}}\cR\circ F^{j} \in \cI^c\}.$$
     \begin{lemma}[Exponential decay of correlations for $q$-point marginals]\label{boundedcov}
For every finite union of open intervals $\cI\subset(0,\infty)$, there is a constant $C>0$  and $\theta\in(0,1)$ such that 
\beq\label{eq:condition D_r(u_n)}
\tmu({A}_{n,1}\cap{\cdots}\cap {A}_{n,q}\cap{A}_{n,q+k+1}\cap\cdots\cap{A}_{n,2q+k})-\(\tmu({A}_{n,1}\cap\cdots\cap {A}_{n,q})\)^2\leq C \theta^k
\eeq
for all $k,n,q\in\N$ satisfying $2q+k\le n$.
Also, there exists $\varepsilon>0$ such that for all $1\le i<j\le n$
\beq\label{eq:condition D_r(u_n)2}
\tmu({A}_{n,i}^c\cap{A}_{n,j}^c)\leq o\(\frac1{n^{1+\varepsilon}}\).
\eeq
     \end{lemma}

\begin{proof}
Note that  $\bI_{A_{n,j}}=\bI_{A_{n,0}}\circ F^j$, and $\bI_{A_{n,0}}\in \cH_{1}$, as it is constant on each level set of  $\cR$. Thus we can apply  \cite{CZ09}-- Theorem 4 for the  first part of this lemma, Ineq. (\ref{eq:condition D_r(u_n)}); the conditions for Theorem 4 of \cite{CZ09} were checked in \cite{Z2016b}. Here we only need to prove (\ref{eq:condition D_r(u_n)2}).

Denote for $\tau>0$, $$\cW_\tau:=\cup_{m=[\tau n^{1/\alpha}]}^{\ff}M_m.$$
Note that $\tilde\mu(\cW_\tau)\sim n^{-1}$.
For any  set $M_m$ in $\cW_\tau$, we foliate it into unstable curves $\{W_\aaa\}$  that stretch completely from one side to the other.  Then by introducing a factor measure on the index of these curves, we define a standard family, denoted as $\cG_\tau=(\cW_\tau,\nu_\tau)$ with the factor measure denoted as $\lambda_{\cG_\tau}$ and $\nu_\tau:=\tilde\mu|_{\cW_\tau}/\tilde\mu(\cW_\tau)$.
One can check that $\cG_\tau$ is a standard family. According to Lemma \ref{Mm}, {$M_m$ and $FM_m$ are strips} that have length $\sim m^{-\frac{\alpha}{\alpha+1}}$, width $\sim m^{-\frac{\alpha^2+\alpha+1}{\alpha+1}}$, and density approximately $1$.
Using the width and density we obtain
\begin{align*}
\cZ(F\cG_\tau)
&=\int_{\cA}|FW_{\aaa}|^{-1}\,\lambda_{\cG_\tau}(d\aaa)\\
&\le C \tilde\mu(\cW_\tau)^{-1}\cdot \sum_{m=\tau n^{1/\alpha}}^{\ff}m^{-\frac{\alpha^2+\alpha+1}{\alpha+1}}\\
&\le C_0  n {\(\frac{1}{n}\)^{\frac{\alpha}{\alpha+1}}}.
\end{align*}
Therefore Lemma \ref{growthlemma} implies that
$$\nu_\tau(r_k(x)<\eps)\leq C_1(\vartheta^{k-1}\cZ(F\cG_\tau)+C_2)\eps^q\leq C_3(\vartheta^{k-1}n^{\frac{1}{\alpha+1}}+C_2)\eps^q$$
which for $k\ge 1$, by the definition of $\nu_\tau$, is equivalent to
\beq\label {Mmrn}\tilde\mu((r_k(x)<\eps)\cap \cW_\tau)\leq  C_3\tilde\mu(\cW_\tau)(\vartheta^{k-1}n^{\frac{1}{\alpha+1}}+C_2)\eps^q.\eeq

We use properties of standard families to finish the proof. For any $\sigma>0$, we know that unstable curves in $\cW_\sigma$ have width less than $\eps= C_4 n^{-\frac{\alpha^2+\alpha+1}{\alpha(\alpha+1)}}$. Thus $$\cW_\sigma\subset \{x:r_0(x)<C_4 n^{-\frac{\alpha^2+\alpha+1}{\alpha(\alpha+1)}}\}$$ which implies that $$\eps^q=(C_4n^{-\frac{\alpha^2+\alpha+1}{\alpha(\alpha+1)}})^{\frac{(\alpha+1)\alpha}{\alpha^2+\alpha+1} }=\cO(n^{-1}).$$
Moreover, \eqref{Mmrn} now gives us
\begin{align}\label{bound2}\nn\tilde\mu(F^{-k}\cW_\sigma\cap \cW_\tau)&\leq \tilde\mu((r_k(x)<\eps)\cap \cW_\tau)\\
&\nn\leq  C_3\tilde\mu(\cW_\tau)(\vartheta^{k-1}n^{\frac{1}{\alpha+1}}+C_2)n^{-1} \\
&\leq C_5\vartheta^{k-1}n^{-\frac{2\alpha+1}{\alpha+1}}+C_6n^{-2}\end{align}
which proves \eqref{eq:condition D_r(u_n)2}.
\end{proof}

\subsection{Poisson convergence for dependent arrays}\label{sec:poisson}
Poisson convergence is the main probabilistic apparatus that allows one to prove convergence to a stable law.
Here we state a Poisson convergence result which is a variant of \cite[Thm 4.1]{adler1978weak} (see also \cite[Thm 5.7.1]{leadbetter1983extremes}). Specifically, it concerns the convergence, to empirical Poisson point processes,  of empirical measures corresponding to stationary triangular arrays of random variables.

It is often useful to think of empirical point processes as integer-valued random measures. Recall that a Poisson random measure or Poisson point process $N(d\lambda)$ with  intensity measure $d\lambda$ (also called the mean measure or control measure) is an independently scattered random measure satisfying, for all Borel sets $B$, $N(\lambda(B))\dsim\Pois(\lambda(B))$ so that
$$\P(N(\lambda(B))=k)=e^{-\lambda(B)}\frac{\lambda(B)^k}{k!}.$$

We start with a lemma due to \cite{kallenberg1973characterization}.
\begin{lemma}[Poisson random measure characterization]
Suppose $(\nu_n)$ is a sequence of random empirical measures on $\R^+$ and that $N(d\lambda)$ is a Poisson random measure on $\R^+$. If for any finite union of open intervals $\cI$ with $\lambda(\cI)<\ff$
\beq\label{ppp1}
\lim_{n\to\ff}\P(\nu_n(\cI)=0)= \exp(-\lambda(\cI))
\ee
and
\beq\label{ppp2}
\lim_{n\to\ff} \E \nu_n(\cI)=\lambda(\cI),
\ee
then we have the convergence in distribution $\nu_n\Rightarrow N(d\lambda)$ under the vague topology on measures (which makes the space of measures a Polish space).
\end{lemma}

\begin{proposition}[Poisson convergence for dependent arrays]\label{prop:1} Suppose $$\{X_{n,j}\,:\,n\geq 1, j=1,\cdots,n\}$$ is a  triangular array of positive random variables such that each row is stationary, and suppose that for some absolutely continuous measure $\lambda$ on $\mathbb{R}^+$,
\begin{align}\label{eq:ppp cond1}
&n\P(X_{n,1}\in\cI)\stackrel{n\to\infty}{\longrightarrow}\lambda(\cI),\quad\text{for any finite union of open intervals }\cI\subset\R^+.
\end{align}
Also suppose that 
for $1\le i<j\leq n$, there exists $\varepsilon>0$ such that
\beq\label{CovXnm}
\P(X_{n,i}\in\cI, X_{n,j}\in\cI)= o\(\frac1{n^{1+\varepsilon}}\).
\eeq
 Finally, suppose that for every finite union of open intervals $\cI$, there exists $\theta\in(0,1)$ and some constant $C>0$ such that for all $k,n,q\in\N$ satisfying $2q+k\le n$,
{\begin{align}\label{cov2}
&| \P(X_{n,1}\in\cI^c, \ldots, X_{n,q}\in\cI^c,X_{n,q+k+1}\in\cI^c,\ldots, X_{n,2q+k}\in\cI^c)\nn\\
&\quad \quad \quad \quad - \(\P(X_{n,1}\in\cI^c, \ldots, X_{n,q}\in\cI^c)\)^2| \le C\theta^k.
\end{align}}

Then we have the following weak convergence of random empirical measures under the vague topology: {$$\sum_{j=1}^n\delta_{X_{n,j}}\Rightarrow N(d\lambda).$$}
\end{proposition}


We remark that conditions (\ref{CovXnm}) and (\ref{cov2}) are analogous to the standard conditions $D'(u_n)$ and $D_r(\textbf{u}_n)$, respectively, in \cite{leadbetter1983extremes}.
These are certainly not the weakest conditions possible, but are tailored to the situation at hand.
\begin{proof}
It is enough to verify \eqref{ppp1} and \eqref{ppp2}. Fix $\cI\subset\R^+$ such that $\lambda(\cI)<\ff$. Since each row of the triangular array is stationary, \eqref{ppp2} follows directly from \eqref{eq:ppp cond1}. So we will concentrate on verifying \eqref{ppp1}.

Condition \eqref{ppp1} would follow if the random variables in each row were independent since then we would have
$$ \P(\nu_n(\cI)=0)={\prod}_{j=1}^n \P(X_{n,j}\notin \cI)\approx \(1-\frac{\lambda(\cI)}{n}\)^n.$$
We do not have independence, but assumptions \eqref{CovXnm} and \eqref{cov2} imply asymptotic independence which
 controls the dependence between $X_{n,1}, \cdots, X_{n,n}$. This will done via Bernstein's classical small-large block method \cite{bernstein1927extension}; see also \cite[Ch. 18]{ibragimov1971independent}. More precisely, choose $$0<a<b<\varepsilon,$$ where $\varepsilon$ is as in \eqref{CovXnm}, and for any $n\geq 1$, divide $\{1,\ldots,n\}$ into a sequence of pairs of alternating big intervals (blocks) of length $[n^b]$ and
small blocks of length $[n^a]$. The number of
pairs of big and small blocks is ${B}= [n/([n^a] + [n^b])]$ so that  $\lim_{n\to\ff} B/n^{1-b}=1$. There may be a leftover partial block $L$ in the end which is is negligible since
$$\sum_{j\in L} \P(X_{n,j}\in \cI)\le \frac{C n^b}{n}.$$
 Thus we may henceforth assume $n=([n^a]+[n^b]){B}$.

We denote by $\cB_{k}$ and $\cS_{k}$ for $k=1,\cdots, {B}$, the elements of $\{1,\ldots,n\}$ in big blocks and small blocks, respectively.
Let
$$Y_{n,k}=\sum_{j\in \cB_{k}} \bI_{\{X_{n,j}\in \cI\}},\,\,\,\,\,\,\,Z_{n,k}=\sum_{j\in \cS_{k}} \bI_{\{X_{n,j}\in \cI\}}.$$
For each $n$, both $\{Y_{n,k}\}$ and $\{Z_{n,k}\}$ are sequences of identically distributed random variables. Let $S'_n=\sum_{k=1}^{B} Y_{n,k}$ and $S''_n=\sum_{k=1}^{B} Z_{n,k}$.

Similar to the argument for $L$, we have
$$\E S''_n \le C\frac{{B} n^{a}}{n}\le C \frac{1}{n^{b-a}},$$
so that we can ignore small blocks. Thus, it is enough to show that $\P(S'_n=0)$ converges to $\exp(-\lambda(\cI))$.
Since big blocks are separated by small blocks, we can peel off one factor at a time in the product $\prod \bI_{\{Y_{n,k}=0\}}$ so that using \eqref{cov2} multiple times implies
\beq\label{Eexpo2}\P(S'_n=0)=(\P(Y_{n,1}=0))^{B}+\cO(\theta_1^{{n^a}}),\eeq
for some $0<\theta_1<1$. It remains to estimate $\P(Y_{n,1}=0)$.
\begin{align*}
\P(Y_{n,1}=0)
&\le 1-\sum_{j=1}^{n^b}\left(\P(X_{n,j}\in \cI)-\sum_{k=j+1}^{n^b}\P(\{X_{n,j}\in \cI\}\cap\{X_{n,k}\in \cI\})\right)\\
&\le 1-\sum_{j=1}^{n^b}\(\P(X_{n,j}\in \cI)-o\(\frac1{n}\)\)\\
&\le 1-n^b\P(X_{n,j}\in \cI)+n^b\,\,\, o\(\frac1{n}\)
\end{align*}
where the second to last inequality follows from \eqref{CovXnm} since $b<\varepsilon$.
An even easier lower bound is given by
\begin{align*}
\P(Y_{n,1}=0)
\ge 1-\sum_{j=1}^{n^b}\P(X_{n,j}\in \cI).
\end{align*}

Putting things together we have
$$\P(S'_n=0)=\(1-{[n^b]}\P(X_{n,1}\in \cI)+o(n^{b-1})\)^{B}+\cO(\theta_1^{n^a}),$$
which is what we need since $B/n^{1-b}\to 1$ and $n\P(X_{n,1}\in \cI)\to \lambda(\cI)$.
\end{proof}

\subsection{Stable laws for return times}\label{sec:R}
Here we prove Theorem \ref{thm:2} by applying Proposition \ref{prop:1} to the (uncentered) triangular array defined by
\begin{align}\label{defnXn} X_{n,k}&:=\frac{1}{\sqrt[\alpha]{n}}\cR\circ F^{k}
\end{align}
for $n\geq 1$ and $1\le k\le n$.
One can easily check that  $\{X_{n,k},k=1,\cdots, n\}$ is a stationary sequence.
The other conditions of
Proposition \ref{prop:1} are verified by  Lemmas \ref{Mm2} and \ref{boundedcov}.

Let
\beq\label{eq:ppp}
\lambda(dx):=\frac{\alpha{\tilde\sigma_\cR}^\alpha dx}{x^{1+\alpha}}
\eeq
where
$$\tilde\sigma_{\cR}^{\alpha}:=\frac{2I_1^\alpha}{\beta\mu(M)|\partial Q| },$$
with $I_1=\int_{0}^{\pi/2}\sin^{\frac{1}{\alpha}}\varphi\,d\varphi$,
is as in Lemma \ref{Mm2}.
In particular, Lemma \ref{Mm2} implies that
{\begin{align*}
\lim_{n\to\infty}n\P(X_{n,0}\ge u)&=\frac{2I_1^\alpha u^{-\alpha}}{\beta\mu(M)|\partial Q|}\\
&=\lambda([u,+\infty)).
\end{align*}}
Thus the random measures
\beq\label{eq:random meas}
\left(\sum_{k=1}^n \delta_{X_{n,k}}, n\in\mathbb{N}\rt)
 \eeq
converge weakly, in the vague topology, to the Poisson random measure with intensity measure $\lambda(dx)$ given in \eqref{eq:ppp}.

Next, for any $\delta>0$, this  implies that the sum of the uncentered terms, $$\sum_{k=1}^n X_{n,k}\bI_{\{\delta^{-1}>X_{n,k}>\delta\}}$$ converges, as $n\to\infty$, to
a compound Poisson distribution.  This follows from vague convergence by simply integrating the identity function $g(x)=x \bI_{\{\delta^{-1}>x>\delta\}}$ over the random measures in \eqref{eq:random meas}, and then taking the limit as $n\to\infty$.
In particular, we obtain the compound Poisson distribution
whose characteristic function has L\'evy exponent
\beq\label{eq:Poisson convergence2}
\tilde\sigma_{\cR}^\alpha\int_{\delta}^{\delta^{-1}} (e^{itx}-1)\alpha x^{-1-\alpha} \, dx.\eeq
This also implies that the sum of the centered terms  $$\sum_{k=1}^n \(X_{n,k}\bI_{\{\delta^{-1}>X_{n,k}>\delta\}}-\E (X_{n,1}\bI_{\{\delta^{-1}>X_{n,k}>\delta\}})\)$$ converge to a distribution with L\'evy exponent
\beq\label{eq:Poisson convergence}
\tilde\sigma_{\cR}^\alpha\int_{\delta}^{\delta^{-1}} \alpha\frac{e^{itx}-1-itx}{x^{1+\alpha}}\, dx.
\eeq
Taking the limit as $\delta\to 0$ for the centered sums with characteristic L\'evy exponent given in \eqref{eq:Poisson convergence},
now proves Theorem \ref{thm:2} since we will see below in Lemma \ref{LHzero} that $$\limsup_{n\to\ff}\sum_{k=1}^n X_{n,k}\bI_{\{X_{n,k}\le \delta\}\cup\{X_{n,k}\ge \delta^{-1}\}}$$ converges in probability to zero as $\delta\to 0$.

Let us remark that it is important to take the limit $\delta\to 0 $ in \eqref{eq:Poisson convergence} rather than in the uncentered sum whose characteristic function
is described in \eqref{eq:Poisson convergence2}, since only \eqref{eq:Poisson convergence}
converges as $\delta\to 0$ (see \cite[Sec 3.7]{durrett2010probability}
 for a thorough discussion).

\section{Intermediate case: induced functions}\label{sec:induced}
In this section we extend the results of the previous section to general induced functions on the space $M$. Recall that
$$
I_f=\frac{1}{4}\int_{0}^{\pi} ( f(r',\varphi)+f(r'',\varphi))\sin^{\frac{1}{\alpha}}\varphi\, d\varphi$$
where $r',r''$ are as in \eqref{r'}.

\begin{theorem}[Stable Limit Theorem for the induced map]\label{thm:3}  Suppose $I_f\neq 0$. Let $ f:\cM\to\mathbb{R}$ satisfy
the assumptions of Theorem 1 and let $\tf$ be the induced function on $M$ constructed by (\ref{inducef}). Then
\beq\label{stable2}
\frac{S_n \tf}{\sqrt[\alpha]{n}}\xrightarrow{d}  S_{\alpha,\tilde\sigma_f}
\eeq
where $\alpha=\frac\beta{\beta-1}$ and  $\tilde\sigma_f^{\alpha}=\frac{2 I_f^{\alpha}}{{\beta\mu(M)|\partial Q|}}$.
\end{theorem}

The idea of the proof is to write a general induced function as
$$\tilde{f}(x) =C\left(\cR(x)-\tmu(\cR)\right)+E(x)$$
where $C$ is a constant and $E(x)$ is an ``error'' function which vanishes in the renormalized limit of Birkhoff sums $n^{-1/\alpha}S_n \tf$.
In order to show that $E(x)$ is inconsequential we will need a decay of correlations result for $\tf$ which is a refined version of Lemma \ref{boundedcov}.
This is much easier to prove if $\tf$ is bounded (for each $n$).
We therefore start this section by truncating both the high and low portions of $\tf$.

Fix a small $\delta>0$, and split $M$ according to the low, intermediate, and high regions of the index $m$
for the sets $\{\cR(x)=m\}$:
\beq
M^L:=\cup_{m<\delta n^{\frac{1}{\alpha}}} M_m,\quad\quad M^I:=\cup_{\delta n^{\frac{1}{\alpha}}\le m< \frac{1}{\delta}n^{\frac{1}{\alpha}}} M_m \quad\text{and}\quad M^H
:=\cup_{m\geq \frac{1}{\delta}n^{\frac{1}{\alpha}}} M_m
\eeq
which all depend implicitly on $n$ and $\delta$.
Note that
$$
\tmu(M^L)\sim 1,\,\,\,\,\,\tmu(M^I)\sim \frac{1}{\delta^{\alpha} n},\,\,\,\,\,\text{ and }\,\,\,\,\,\tmu(M^H)\sim \frac{\delta^{\alpha}}{n},
$$
where the $\delta^\alpha$ is irrelevant at this point since the relation $\sim$ is up to a constant as $n\to\infty$.  Nevertheless we have written the factor of $\delta^\alpha$ since we will eventually send $\delta\to 0$ (after we compare quantities which have equivalent asymptotics as $n\to\infty$).
We also put
\beq\label{function decomp}
\tf_{L}:=\tf|_{M^{L}}, \quad\tf_{I}:=\tf|_{M^{I}} \quad\text{and} \quad\tf_H:=\tf|_{M^{H}}
\eeq
so that
$
\tf=\tf_L+\tf_I+\tf_H.$
Note that the mean $\tmu(\tf)=0$ does not imply  that the truncated means are also zero.  However,
{\begin{align*}
0&=\tmu(\tf)=\tmu(\tf_L)+\tmu(\tf_I )+\tmu(\tf_H)\\
\end{align*}
and so, for each fixed $\delta$, we can always subtract these three constants (which sum to zero) in order to center $\tf_L, \tf_I$, and $\tf_H$. Thus we assume without loss of generality that
 $$\tmu(\tf_L)=\tmu(\tf_I )=\tmu(\tf_H)=0.$$}


\subsection{Decay of correlations for induced functions}

Fix $\delta>0$ and let  $\cW_{i,j}:=\cup_{m=i}^j M_m$ be the union of cells with indices satisfying $i\leq j\leq \frac{1}{\delta} n^{\frac{1}{\alpha}}$ so that $\cW_{i,j}\subset M^L\cup M^I$.
Note that the viable range of $i$ and $j$ depends on $n$ and $\delta$.
Denote
\beq\nn \label{eq: function sequence}
{\tf}_{i,j}:=\tf\cdot \bI_{\cW_{i,j}}-\tmu(\tf\cdot\bI_{\cW_{i,j}}).
\eeq
We also set $\alpha_0=\frac{1+\sqrt{5}}{2}$.

\begin{proposition}[Exponential decay of correlations for $\tf_{i,j}$]\label{boundedcovtf}
 For $1\le i\leq j\leq \frac{1}{\delta} n^{\frac{1}{\alpha}}$, {$k\ge 1$}, and $\alpha\in [\alpha_0,2)$, we have
\beq\label{eq:condition D_r(u_n)2f}
|\tmu(\tf_{i,j}\circ F^k\cdot{\tf}_{i,j})|\leq C_f\theta^k,
\eeq
where  $\theta\in (0,1)$ is a constant, and $C_f>0$ is a constant that depends on $\|f\|_{C^{\gamma}}$.
On the other hand, for any $\alpha\in (1,\alpha_0)$, there exists $\chi>0$, such that for any $1\le k\leq \chi \ln n$,
\beq\label{eq:new condition}
|\tmu(\tf_{i,j}\circ F^k\cdot{\tf}_{i,j})|\leq C_f \delta^{-1-\frac{\alpha}{\alpha+1}+\alpha}n^{\frac{2}{\alpha}-1-\frac{1}{\alpha(\alpha+1)}}\theta^k,
\eeq
and for any $k\geq \chi \ln n$, (\ref{eq:condition D_r(u_n)2f}) holds.
\end{proposition}

\begin{proof}
For any  set $M_m$ in $\cW_{i,j}$, we foliate it into unstable curves  that stretch completely from one side to the other. Let $\{W_{\aaa}, \aaa\in \cA, \lambda\}$ be the foliation, and $\lambda$ the factor measure defined on the index set $\cA$.  This enables us to define a standard family, denoted as $\cG_{i,j}=(\cW_{i,j},\tmu_{i,j})$, where $\tmu_{i,j}:=\tmu|_{\cW_{i,j}}$. Moreover, we also define $\cG_{m}=(M_{m},\tmu|_{M_m})$.

Our first step in proving the decay of correlations is to investigate the $\cZ$ function of $F^k\cG_{i,j}$.
According to Lemma \ref{Mm}, $FM_m$ can be approximated by a strip that has length $\sim m^{-\frac{\alpha}{\alpha+1}}$ and width $\sim m^{-\frac{\alpha^2+\alpha+1}{\alpha+1}}$. Also, by construction, the density of $\tmu_{i,j}$ is of order $1$ on $ M_m$. Recalling that $\cZ(F\cG_{m})$ is the average inverse length of unstable curves in $FM_{m}$, by \eqref{cZ}  we obtain
\begin{align*}
\cZ(F\cG_{m})&=\tilde\mu(M_{m})^{-1}\int_{\aaa\in\cA}|W_{\aaa}|^{-1}\,\lambda_{\cG_{m}}(d\aaa)\\
&\le C_1\tilde\mu(M_{m})^{-1}\cdot m^{\frac{\alpha}{\alpha+1}}\cdot \tilde\mu(M_{m})\\
&\leq C m^{-\frac{\alpha}{\alpha+1}}.\end{align*}

Similarly, for $j<\infty$,
\begin{align*}
\cZ(F\cG_{i,j})&=\tilde\mu(\cW_{i,j})^{-1}\int_{\aaa\in\cA}|W_{\aaa}|^{-1}\,\lambda_{\cG_{i,j}}(d\aaa)\\
&\le C_1\tilde\mu(\cW_{i,j})^{-1}\cdot \sum_{m=i}^{j} m^{\frac{\alpha}{\alpha+1}}\cdot m^{-\alpha-1}\\
&= C i^{\frac{\alpha}{\alpha+1}}.\end{align*}

For $m<l$ we have
$$\tmu(F(\cW_{i,j})\cap \cW_{m,l})\leq F_*\tmu_{i,j}(r<\eps_m)$$
where $\eps_m=C_0m^{-\frac{\alpha^2+\alpha+1}{\alpha+1}}$ for some $C_0>0$, i.e., the order of the width of the largest cell in $\cW_{m,l}$.
Using  Lemma \ref{growthlemma} with $q$ as in \eqref{q}, we have that {for $k\ge 1$}
\begin{align}\label{estMiMk}
\tmu(F^k(\cW_{i,j})\cap \cW_{m,l})&\leq F^k_*\tmu_{i,j}(r<\eps_m)\nonumber\\
&\leq C'(\vartheta^{k-1}\cZ(F\cG_{i,j})+C'')\eps_m ^q\tmu(\cW_{i,j})\nonumber\\
&\leq C\vartheta^{k}(i^{\frac{\alpha}{\alpha+1}}+C'')m^{-\alpha} i^{-\alpha}\nonumber\\
&= C\vartheta^{k}i^{\frac{\alpha}{\alpha+1}-\alpha}+CC''i^{-\alpha}m^{-\alpha}
\end{align}

For any fixed large $k$, we truncate $\tf_{i,j}$ at one extra level $p$, with $i\leq p\leq j$, which will be chosen later. Since
$\tf_{i,j}=\tf|_{\cW_{i,j}}-\tilde\mu(\tf|_{\cW_{i,j}})$, we have
\begin{align*}
\tmu(\tf_{i,j}\circ F^k\cdot \tf_{i,j})=\tmu(\tf|_{\cW_{i,j}}\circ F^k \cdot \tf|_{\cW_{i,j}})-\tmu(\tf|_{\cW_{i,j}})^2.\end{align*}
Also define another decomposition
\begin{align}\label{tildefijtrun}\tf_{i,j}&=g_{i,p}+g_{p,j}\nonumber\\
&:=\left(\tf \bI_{\cW_{i,p}}-\tilde\mu(\tf\bI_{\cW_{i,j}})\bI_{\cW_{1,p}}\right)+\left(\tf\bI_{\cW_{p,j}}-\tilde\mu(\tf\bI_{\cW_{i,j}}) \bI_{\cW_{p,\infty}}\right)\end{align}
where the right-side functions are not mean-zero in general (but their sum is).

The function $g_{i,p}$ satisfies $\|g_{i,p}\|_{\infty}\leq \|f\|_{\infty} p$, and its H\"{o}lder norm satisfies $$\|g_{i,p}\|_{\gamma}\leq C\|f\|_{\gamma} p^{1+\gamma},$$ by Lemma \ref{Holdertf}. Thus by (\ref{correT}), we know that
\beq\label{iqiq}
\tmu(g_{i,p}\circ F^k, g_{i,p})\leq C\|g_{i,p}\|_{C^{\gamma}}^2 \vartheta^k +\tmu(g_{i,p})^2\leq C p^{2+2\gamma}\vartheta^k +\cO(p^{2-2\alpha})
\eeq
where we have used Lemma \ref{Mm} (7) in order to estimate
$$\tmu(g_{i,p})=-\tmu(g_{p,j})=\cO(p^{1-\alpha}), \,\,\,\tmu(\tf \cdot \bI_{\cW_{p,j}})=\cO(p^{1-\alpha}).$$
Moreover,
\begin{align}\label{iqiq2}
\tmu(g_{i,p}, \tilde\mu(\tf\bI_{\cW_{i,j}}) \bI_{\cW_{p,\infty}}\circ F^k)&\leq \tilde\mu(\tf\bI_{\cW_{i,j}})\left(C\|g_{i,p}\|_{C^{\gamma}} \vartheta^k +\tmu(g_{i,p})\right)\nonumber\\
&\leq C p^{2-\alpha+\gamma}\vartheta^k +\cO(p^{2-2\alpha}).
\end{align}
Similarly,
\begin{align}\label{iqiq3}
\tmu(\tilde\mu(\tf\bI_{\cW_{i,j}}) \bI_{\cW_{p,\infty}}, \tilde\mu(\tf\bI_{\cW_{i,j}}) \bI_{\cW_{p,\infty}}\circ F^k)&\leq \tilde\mu(\tf\bI_{\cW_{i,j}})^2\left(C\vartheta^k +1\right).
\end{align}

We now estimate
\begin{align*}
&\tmu(g_{i,p},\tf_{\cW_{p,j}}\circ F^k)\leq C\|f\|_{\infty}^2\sum_{m=1}^p\sum_{l=p}^j m\cdot l\cdot \tmu(M_l\cap F^k M_m)\\
&=C\|f\|_{\infty}^2\sum_{m=1}^p\sum_{l=p}^j \sum_{t=1}^m \sum_{s=1}^l \tmu(M_l\cap F^k M_m)\\
&=C\|f\|_{\infty}^2\sum_{m=1}^p\sum_{t=1}^m
\left(\sum_{s=1}^p \sum_{l=p}^j  \tmu(M_l\cap F^k M_m)+\sum_{s=p}^j \sum_{l=s}^j  \tmu(M_l\cap F^k M_m)\right)\\
&=C\|f\|_{\infty}^2\sum_{m=1}^p\sum_{t=1}^m
\left(p\cdot \tmu(\cW_{p,j}\cap F^k M_m)+\sum_{s=p}^j   \tmu(\cW_{s,j}\cap F^k M_m)\right)\\
&=C\|f\|_{\infty}^2\sum_{t=1}^p
\left(p\cdot \tmu(\cW_{p,j}\cap F^k \cW_{t,p})+\sum_{s=p}^j   \tmu(\cW_{s,j}\cap F^k \cW_{t,p})\right).
\end{align*}
This implies by \eqref{estMiMk} that
\begin{align*}
&\tmu(g_{i,p},\tf_{\cW_{p,j}}\circ F^k)\\
&\leq C\|f\|_{\infty}^2\sum_{t=1}^p
\left((\vartheta^{k-1}{t^{\frac{\alpha}{\alpha+1}-\alpha}}+C''t^{-\alpha})p^{1-\alpha} +\sum_{s=p}^j  (\vartheta^{k-1}{t^{\frac{\alpha}{\alpha+1}-\alpha}}+C''t^{-\alpha})s^{-\alpha} \right)\\
\end{align*}
We have that $\alpha_0=\frac{1+\sqrt{5}}{2}$ satisfies $\alpha_0=1+\frac{\alpha_0}{\alpha_0+1}$.
Then combining with (\ref{iqiq2}), one can check that for $\alpha\in (1,\alpha_0)$,
\beq\label{lesalpha}
\tmu(g_{i,p},g_{p,j}\circ F^k)\le C_f(p^{1-\alpha}(\vartheta^k p^{1+\frac{\alpha}{\alpha+1}-\alpha})).
\eeq
On the other hand, for $\alpha\in [\alpha_0,2)$,
 one can check that
\beq\label{lesalpha0}
\tmu(g_{i,p},g_{p,j}\circ F^k)\le C_f p^{1-\alpha}\vartheta^k .
\eeq
Similarly, we can show that, for $\alpha\in (1,\alpha_0)$,
\beq\label{lesalpha00}
\tmu(g_{p,j},g_{i,p}\circ F^k)\le C_f((\vartheta^k j^{1+\frac{\alpha}{\alpha+1}-\alpha}+p^{1-\alpha}))
\eeq
while for $\alpha\in [\alpha_0,2)$,
\beq\label{lesalpha000}
\tmu(g_{p,j},g_{i,p}\circ F^k)\le C_f((\vartheta^k p^{1+\frac{\alpha}{\alpha+1}-\alpha}+p^{1-\alpha})).
\eeq

Next, we estimate
\begin{align*}
&\tmu(\tf_{\cW_{p,j}},\tf_{\cW_{p,j}}\circ F^k)\leq C\|f\|_{\infty}^2\sum_{m=p}^j\sum_{l=p}^j m\cdot l\cdot \tmu(M_l\cap F^k M_m)\\
&=C\|f\|_{\infty}^2\sum_{t=1}^p
\left(p\cdot \tmu(\cW_{p,j}\cap F^k \cW_{p,j})+\sum_{s=p}^j   \tmu(\cW_{s,j}\cap F^k \cW_{p,j})\right)\\
&+C\|f\|_{\infty}^2\sum_{t=p}^j
\left(p\cdot \tmu(\cW_{p,j}\cap F^k \cW_{t,j})+\sum_{s=p}^j   \tmu(\cW_{s,j}\cap F^k \cW_{t,j})\right)\\
&= C\|f\|_{\infty}^2
\left( (\vartheta^{k-1}{p^{\frac{\alpha}{\alpha+1}-\alpha}}+C''p^{-\alpha})p^{1-\alpha} +\sum_{s=p}^j (\vartheta^{k-1}{p^{\frac{\alpha}{\alpha+1}-\alpha}}+C''p^{-\alpha})s^{-\alpha} \right)\\
&+C\|f\|_{\infty}^2\sum_{t=p}^j
\left((\vartheta^{k-1}{t^{\frac{\alpha}{\alpha+1}-\alpha}}+C''t^{-\alpha})p^{1-\alpha} +\sum_{s=p}^j  (\vartheta^{k-1}{t^{\frac{\alpha}{\alpha+1}-\alpha}}+C''t^{-\alpha})s^{-\alpha} \right).
\end{align*}

Combining with (\ref{iqiq3}), we obtain for $\alpha\in (1,\alpha_0)$,
\beq\label{lesalpha06}
\tmu(g_{p,j},g_{p,j}\circ F^k)\le C_f(p^{1-\alpha}(\vartheta^k j^{1+\frac{\alpha}{\alpha+1}-\alpha}+ p^{1-\alpha})).
\ee
On the other hand, for $\alpha\in [\alpha_0,2)$,
\beq\label{lesalpha07}
\tmu(g_{p,j},g_{p,j}\circ F^k)\le C_f(\vartheta^k p^{1-\alpha} p^{1+\frac{\alpha}{\alpha+1}-\alpha}+p^{2(1-\alpha)}).
\eeq

Combining the above estimations, we have for $\alpha\in (1,\alpha_0)$,\begin{align*}
\tmu(\tf_{i,j},\tf_{i,j}\circ F^k)&\leq C_f\(\vartheta^k p^{1-\alpha} j^{1+\frac{\alpha}{\alpha+1}-\alpha}+ \vartheta^k p^{2+2\gamma} + p^{1-\alpha}\)\end{align*}
while for $\alpha\in [\alpha_0,2)$,
 \begin{align*}
\tmu(\tf_{i,j},\tf_{i,j}\circ F^k)&\leq C_f\(\vartheta^k p^{1-\alpha}+ \vartheta^k p^{2+2\gamma} +p^{1-\alpha}\).\end{align*}

Consider $k$ such that $$\vartheta^{-\frac{k}{3(1+\gamma)}}<{\|f\|_{\infty}^{\frac{\alpha+1}{\alpha}}}\frac{1}{\delta} n^{\frac{1}{\alpha}},$$ where the exponent of $\|f\|_\infty$ comes from \eqref{defnC'} below).
For $\alpha\in (1,\alpha_0)$, we choose $p=\vartheta^{-\frac{k}{3(1+\gamma)}}$.
{The above estimations imply that
$$\tmu(\tf_{i,j},\tf_{i,j}\circ F^k)\leq C_f \delta^{-1-\frac{\alpha}{\alpha+1}+\alpha} n^{\frac{2}{\alpha}-1-\frac{1}{\alpha(\alpha+1)}}\theta^k$$
{where $\theta=\vartheta^{\frac{\alpha-1}{3(1+\gamma)}}$} and where we have used that $j\leq \frac{1}{\delta}n^{\frac{1}{\alpha}}$.
For $\alpha\in [\alpha_0,2)$, we also choose $p=\vartheta^{-\frac{k}{3(1+\gamma)}}$.
Then the above estimations imply that
$$\tmu(\tf_{i,j},\tf_{i,j}\circ F^k)\leq C_f\theta^k.$$}

Next, we consider the case where $k$ is such that
\beq\label{properk}
{\vartheta^{-\frac{k}{3(1+\gamma)}}}\geq\|f\|_{\infty}^{\frac{\alpha+1}{\alpha}}\frac{1}{\delta} n^{\frac{1}{\alpha}}.
\eeq
We assume $$h_{i,j}:=\tf\cdot \bI_{\cW_{i,j}}\geq 0,$$ otherwise, we decompose
$h_{i,j}=h^+_{i,j}-h^-_{i,j}$ into the difference of its positive and negative parts.

Define a standard family, denoted as $\tilde\cG_{i,j}=(\cW_{i,j},\nu_{i,j})$, with $d\nu_{i,j}=h_{i,j}d\mu$.
Our first step  is to show that for  $k$ satisfying \eqref{properk}, $F^k\tilde\cG_{i,j}$ is a proper family.
According to Lemma \ref{Mm}, $FM_m$ is approximated by a strip that has length $\sim m^{-\frac{\alpha}{\alpha+1}}$ and width $\sim m^{-\frac{\alpha^2+\alpha+1}{\alpha+1}}$. Also, by construction, the density of $\nu_{i,j}$ is of order $m$ on $ FM_m$. Thus we obtain for $j<\infty$,
\begin{align*}
\cZ(F\tilde\cG_{i,j})&=\tilde\mu(h_{i,j})^{-1}\int_{\aaa\in\cA}|FW_{\aaa}|^{-1}\,\lambda_{\tilde\cG_{i,j}}(d\aaa)\\
\\
&\le C
\|f\|_{\infty}\tilde\mu(h_{i,j})^{-1}\cdot \sum_{m=i}^{j} m\cdot m^{-\frac{\alpha^2+\alpha+1}{\alpha+1}}\\
&= C\|f\|_{\infty}\tilde\mu(h_{i,j})^{-1}\sum_{m=i}^{j}  {m^{-\frac{\alpha^2}{\alpha+1}}}.\end{align*}
Note that for $\alpha\in (1,\alpha_0)$, we have $\frac{\alpha^2}{\alpha+1}<1$, which implies that
$$ \cZ(F\tilde\cG_{i,j})\leq C\|f\|_{\infty} j^{1-\frac{\alpha^2}{\alpha+1}}i^{\alpha-1}\leq C\|f\|_{\infty} j^{\alpha-\frac{\alpha^2}{\alpha+1}}\leq C'\|f\|_{\infty} n^{\frac{1}{\alpha+1}}.$$
Here we used the following fact:
$$\tilde\mu(h_{i,j})= \sum_{m=i}^j \mu(\tf\cdot \bI_{M_m})\sim\sum_{m=i}^j m\cdot \tmu(\bI_{M_m})\sim i^{1-\alpha}.$$
If $\alpha\in [\alpha_0,2)$, then
 we have $\frac{\alpha^2}{\alpha+1}\geq 1$, which also implies that
$$\cZ(F\tilde\cG_{i,j})\leq C\|f\|_{\infty} i^{\alpha-\frac{\alpha^2}{\alpha+1}}\leq C'\|f\|_{\infty} n^{\frac{1}{\alpha+1}}.$$

Using the inequality $$\cZ(F^k\tilde\cG_{i,j})\leq C_1(\vartheta^{k-1}\cZ(F\tilde\cG_{i,j})+C_2)$$ from Lemma \ref{growthlemma} combined with $\cZ(F\tilde\cG_{i,j})\leq C'\|f\|_{\infty} n^{\frac{1}{\alpha+1}}$, we have that for \mbox{$\|f\|_{\infty}^{\frac{\alpha+1}{\alpha}}\frac{1}{\delta} n^{\frac{1}{\alpha}}\le \vartheta^{-\frac{k}{3(1+\gamma)}}$,}
\begin{align}\label{defnC'}\cZ(F^k\tilde\cG_{i,j})&\leq C_1(\vartheta^{k-1}\cZ(F\tilde\cG_{i,j})+C_2)\nn\\
&\leq C_1(C'\vartheta^{k-1}\|f\|_{\ff}n^{\frac{1}{\aa+1}}+C_2)\nn\\
&{\leq C_3(\delta^{\frac{\alpha}{\aa+1}}{\vartheta^{k(1-\frac{\alpha}{3(\alpha+1)(1+\gamma)})-1}}+C_4)}.\end{align}
Using our freedom to choose
$C_\text{prop}$, we  choose it such that \beq\label{cq}
C_\text{prop}>C_3C_4.
\eeq
Since we can always take $\delta<1$, this implies that $F^k\tilde\cG_{i,j}$ is  a proper standard family.
By Lemma \ref{equidistribution},
since $\tilde\cG_{i,j}$ is a proper family, \beq\label{properg12}
|F^k\nu(\tf_{i,j})|=|F^k\nu(h_{i,j})-\tilde\mu(h_{i,j})|\leq C \|h_{i,j}\|_{C^{\gamma}}\vartheta^k.
\eeq
We know that $\|h_{i,j}\|_{\infty}\leq \|f\|_{\infty} j$, and Lemma \ref{Holdertf} implies that the H\"{o}lder norm of $h_{i,j}=\cO(j^{1+\gamma})$. Thus we have by \eqref{properk}
$$
|F^k\nu(\tf_{i,j})|\leq C_0(C_1j+ C_2 j^{1+\gamma})\vartheta^k{\leq{C_f\delta}^{-1-\gamma}n^{\frac{1+\gamma}{\alpha}} \vartheta^k\leq C_f\vartheta^{2k/3}.}$$
Set $\theta=\vartheta^{2/3}$ so that
\begin{align*}
\tmu(\tf_{i,j}\circ F^k, \tf_{i,j})&\leq  C_f\theta^{k}.
\end{align*}
\end{proof}

\subsection{Stable laws for induced functions}

Our next lemma says that $\tf$ is essentially determined by $\tf_I$.
\begin{lemma}[Vanishing of truncated portions]\label{LHzero} Let $\tf$ be induced by a H\"older continuous $f\in L_{\infty}(\cM,\mu)$. Recalling that the truncations $\tf_L, \tf_H$ depend on $\delta$ and $n$, we have
$$\lim_{\delta\to 0}\frac{S_n\tf_L}{n^{\frac{1}{\alpha}}}=\lim_{\delta\to 0}\frac{S_n\tf_H}{n^{\frac{1}{\alpha}}}=0$$  in probability, uniformly in $n$.\end{lemma}
\begin{proof}
Let us start with the easier assertion for $\tf_H$. By Lemma \ref{Mm}
part (7),
$\tmu(M^H)\le C{\delta^{\alpha}}/{n}$ so that
\beq\label{MH2}
\tmu(\{x:|\tf_{H}\circ   F^m|\ge \delta^{-1}n^{1/\alpha} \text{ for some }m\le n\})\le C \delta^{\alpha},
\ee which implies the claim for $\tf_H$.

Concerning the lower truncation, note that
$$\tmu\(\left(n^{-\frac{1}{\alpha}}S_n\tf_L\right)^2\)\le \tmu\(\left(\frac{\|f\|_\ff}{n^{\frac{1}{\alpha}}}S_n\cR_L\right)^2\).$$
where $\cR_L$ is $\cR$ truncated according to \eqref{function decomp}.
By Lemma \ref{Mm} part (7),
$$\tmu\(\cR_L^2\)\le
C{\sum_{k=1}^{\lfloor\delta n^{1/\alpha}\rfloor}} k^2 k^{-1-\alpha} \le C\delta^{2-\aa} n^{\frac{2}{\aa}-1}.$$

Moreover, by Proposition \ref{boundedcovtf}, we have that for $k\ge 1$ and $1<\alpha<\frac{1+\sqrt{5}}{2}$,
$$\text{Cov}(\tf_L, \tf_L\circ   F^k)\leq   C_f\delta^{1+\frac{\alpha}{\alpha+1}-\alpha}n^{\frac{2}{\alpha}-1-\frac{1}{\alpha(\alpha+1)}}\theta^k,$$
and  for $k\ge 1$ and $\frac{1+\sqrt{5}}{2}\le \alpha<2$,
$$\text{Cov}(\tf_L, \tf_L\circ   F^k)\leq   C_f\theta^k.$$

For all $n$ and some $\eps>0$,
\begin{align*}\text{Var}(S_n\tf_L )&\le 2\sum_{k=0}^{n-1} (n-k) \text{Cov}(\tf_L, \tf_L\circ   F^k)\\
&\leq C {\delta^\eps}n^{\frac{2}{\alpha}}\end{align*}
Thus the variance of ${n^{-\frac{1}{\alpha}}}S_n\tf_L$ is bounded, and taking $\delta\to 0$ proves the claim.
\end{proof}

It is thus enough to  show that {$n^{-\frac{1}{\alpha}}S_n\tf_I$} converges  in distribution to our desired stable law as, firstly, $n\to\infty$ and then, secondly, $\delta\to 0$. 

 Let $\cR_I:=\cR|_{M^I}$ and define $E_I$ by writing $\tf_I$  as:
{\beq\label{defnIfI}\tilde{f}_I =:\frac{I_f}{I_1}\(\cR_I-\tmu(\cR_I)\)+E_I\eeq}
 where  as before {$I_f={\frac14 \int_{0}^{\pi} ( f(r',\varphi)+f(r'',\varphi))\sin^{\frac{1}{\alpha}}\varphi\, d\varphi}.$}
 Note that $\tmu(E_I)=0$  since $\tf_I$ has zero mean. We will prove the following estimate on the error term in Section \ref{sec:scale}:
\begin{lemma} \label{E bound}
For some  constant $C>0$, any {$x\in M^I$,}
\beq\nn
|E_I(x)|\leq C \|f\|_{\ff}|\cR_I-\tmu(\cR_I)|^{1-\frac{\gamma}{\beta-1} }.
\eeq
\end{lemma}

\bigskip

\begin{proof}[Proof of Theorem \ref{thm:3}]
Approximating the measure of $M^I$ as $\cO(1/n)$ using Lemma \ref{Mm} part (7), and combining this with  Lemma \ref{E bound}, we have that $$\text{Var}(E_I)=\tmu(E_I^2)=\cO(n^{\frac{2}{\alpha}-1-\frac{2\gamma(\alpha-1)}{\alpha}} )$$
where we have used that $\beta-1=(\alpha-1)^{-1}$.
By Proposition \ref{boundedcovtf}, for $k\ge 1$,
$$\text{Cov}(E_I, E_I\circ   F^k)\leq   Cn^{\frac{2}{\alpha}-1-\frac{1}{\alpha(\alpha+1)}} \theta^k.$$
Thus, for all $n$,
\begin{align*}\text{Var}(S_n E_I )&\le 2\sum_{k=0}^{n-1} (n-k) \text{Cov}( E_I, E_I\circ   F^k)\\
&\leq C n^{\frac{2}{\alpha}-\epsilon}\end{align*}
for some $\epsilon>0$.
We  therefore have that,
\beq\label{Enkdecay}
 \frac{S_n E_I}{\sqrt[\alpha]{n}}\to 0
\eeq
in $L_2$, as $n\to\infty$.

Combining \eqref{Enkdecay} with (\ref{defnIfI}), we have that
$$\lim_{n\to\ff}\sum_{k=1}^n\frac{\tilde f_I\circ F^k(x)}{\sqrt[\alpha]{n}} =\lim_{n\to\ff}\frac{I_f}{I_1}\sum_{k=1}^n\frac{\cR_I\circ F^k(x)-\tmu(\cR_I)}{\sqrt[\alpha]{n}}.$$
The limit on the right side was found in Subsection \ref{sec:R}, thus we have that
$$\lim_{\delta\to 0}\lim_{n\to\ff}\sum_{k=1}^n\frac{\tilde f_I\circ F^k(x)}{\sqrt[\alpha]{n}} ={S_{\alpha,\tilde\sigma_f}}$$
with $\tilde\sigma_f=\frac{I_f}{I_1} \tilde\sigma_{\cR}$.
\end{proof}

\section{ Proof of the main theorem}\label{sec:main thm}

In this section we prove that Theorem \ref{thm:1} follows from Theorem \ref{thm:3}. For this we will utilize  weak convergence of processes in the space of  cadlag functions, $D([0,2])$.

Recall that $\tf_I$ is defined in terms of $\delta$.  We do not assume $\tmu(f_I)=0$ here.
The first thing to notice is that for each fixed $\delta>0$, we have that $$\left(n^{-\frac{1}{\alpha}}S_{[nt]}\tf_I, t\in[0,2]\right)$$ weakly converges in Skorokhod's $J_1$ topology, as $n\to\infty$, to a compound Poisson process, i.e., a L\'evy process with characteristic exponent \eqref{eq:Poisson convergence2}
(see \cite{billingsley1999convergence} for a discussion of the $J_1$ topology and \cite{kyprianou2006introductory} for a discussion of L\'evy processes). This follows because, for $\delta>0$, the intensity measure of the limiting Poisson point process is finite:
\beq
\int_\delta^{1/\delta}x^{-1-\alpha}\,dx<\infty.
\ee
Therefore, there are a bounded (in distribution) number of nonzero terms in the summation. In other words, there is a random variable $X$ such that the number of nonzero terms in $S_{[nt]}\tf_I$ is stochastically bounded by $X$ for all $n$.

This implies that the centered sums
\beq\label{eq:flim}
\(n^{-\frac{1}{\alpha}}S_{[nt]}(\tf_I-\tmu(\tf_I)), t\in[0,2]\right)
\eeq
also converge in the $J_1$ topology.

Next, let $\cN(n,x)$ be the function which counts the number of returns to $M$ up to time $n$, with respect to the map $T$, starting from $x\in M$.
By the Ergodic Theorem, $$\tmu(\cR)\frac{\cN(n,x)}{n}\to 1$$ as $n\to\infty$ for $\tmu$-a.s. $x\in M$. Therefore, the following sequence of functions over the time interval $t\in[0,2]$ converges in $C([0,2])$ in the uniform topology ($\tmu$-a.s., and thus weakly):
\begin{equation}\label{eq:yn}
y_n(t)=\begin{cases}\frac{\cN(n,x)\tmu(\cR)}{n}t, \ t\in[0,1]\\
2+\left(2-\frac{\cN(n,x)\tmu(\cR)}{n}\rt)\left(t-2\rt), \ t\in[1,2].
\end{cases}
\end{equation}
The sequence $(y_n)$ in fact converges to the identity function, and it is not difficult to see that the  composition map  $(x,y) \to x \circ y\,$ is continuous from $D([0,2]) \times C([0,2])$ to $D([0,2])$.

By the Continuous Mapping Theorem we obtain that for each fixed $\delta>0$,
the limiting normalized sums
\begin{align}\label{eq:random sums}
&\lim_{n\to\ff}\(\tmu(\cR)\cN\)^{-\frac{1}{\alpha}}\left(S_{[\tmu(\cR)\cN t]}(\tf_I-\tmu(\tf_I))\right)
\end{align}
weakly converge in $J_1$ (where $t\in[0,2]$) to the same limit as that of \eqref{eq:flim}.
Similarly since
$$\tmu(\cR)\frac{\cN(n,x)+1}{n}\to 1$$
as $n\to\infty$ for $\tmu$-a.s. $x\in M$,
we can conclude that \eqref{eq:random sums} with $\cN+1$ replacing $\cN$ also converges to the same limit as that of \eqref{eq:flim}. Moreover, the same convergence holds for any $\cN'([nt],x)$ between $\cN$ and $\cN+1$.

To complete the proof, we simply consider the marginal distribution at $t=1/\tmu(\cR)=\mu(M)\in (0,1)$
and appeal to Lemma \ref{LHzero}  which allows us to take the limit as $\delta\to 0$.
This gives us that
\begin{align}\label{eq:random sums2}
\lim_{n\to\ff}\frac{S_{\cN}\tf}{\cN^{\frac{1}{\alpha}}}
&=\lim_{n\to\ff}\frac{\cS_{n}f}{\cN^{\frac{1}{\alpha}}}\frac{\cN^{\frac{1}{\alpha}}}{(\mu(M)n)^{\frac{1}{\alpha}}}
\end{align}
converges in distribution to $S_{\alpha,\tilde\sigma}$.

\vspace{5mm}
\noindent{\bf Remark}: Note that the convergence of \eqref{eq:random sums} is a process-level convergence in $J_1$; however, this corresponds to the truncated function $\tf_I$ for $\delta>0$. In order to obtain convergence in $J_1$ for the  full normalized summation, one is required to show $$\left([\cN t]^{-\frac{1}{\alpha}}S_{[\cN t]}(\tf_L-\tmu(\tf_L)), t\in[0,2]\right)$$ converges to the $0$ process in $J_1$ which is presumably quite difficult (the analogous result for the ``high'' truncation would also be needed, but would follow from Lemma \ref{LHzero}).

\section{Proofs of Lemmas \ref{Mm2} and \ref{E bound}: the scale parameter}\label{sec:scale}

Here we calculate the parameter
\beq\label{eq:scale parameter2}\tilde\sigma_f^\alpha:=\lim_{n\to\infty} n\varepsilon^\alpha\tmu\left({|\tf|} >\varepsilon\, n^{\frac{1}{\alpha}}\right)=\frac{2 I_f^{\alpha}}{\beta\mu(M)|\partial Q|}
\eeq
which easily implies
$$
\quad\sigma_f=\(\mu(M)\tilde\sigma_f^\alpha\rt)^{1/\alpha}.$$
We start with the special case of the return time function where {$I_f=I_1$} as in \eqref{eq:scale parameter}.

\begin{proof}[Proof of Lemma \ref{Mm2}]
Note that $M$ consists of all collisions on $\partial Q\setminus(\Gamma_1\cup\Gamma_2)$, thus the measure of $M$ satisfies
\beq\nn
\mu(M)=\frac{1}{2 |\partial Q|}\int_0^{|\partial Q|-|\Gamma_1|-|\Gamma_2|}\int_{0}^{\pi} \sin\varphi\,d\varphi\,dr=\frac{|\partial Q|-|\Gamma_1|-|\Gamma_2|}{|\partial Q|}
\eeq

To estimate the constant $I_1$ in {(\ref{eq:scale parameter})}, we need to estimate the measure of the set $\{\cR\geq n\}$ as a subset in $\cM$; although $\cR(x)$ is originally defined on $M$, it can be extended to the entire space $\cM$ as the first hitting time function.
Recall that we denote $ F ^n x=(r_n,\varphi_n)$ and $\eta_n:=\min(\varphi_n, \pi-\varphi_n)$. Lemma \ref{Mm} and \eqref{sn} imply that for any $n\in [1,\cdots, N_2]$, $s_n$ and $\eta_n$ have the following relation:
{\beq\label{snphin}s_n^{\beta-1}\sim \frac{\eta_n}{n}\eeq}
According to \cite[Proposition 2 part (6)]{Z2016b}, we know that for $N$ large enough and for any $x\in M_N$, the following sequence is almost constant for $n=1,\ldots, N$:
\beq\label{hn} H_n:=s_n^{\beta} \sin \eta_n=C_N+\cO(s_n^{2\beta-1})\eeq
where $C_N=c_h N^{-\alpha} $ is a constant which depends only on $N$, and $c_h>0$ is a constant. 

 In order to estimate $c_h$, we use an elliptic integral and introduce
\beq\nn v_n:=\int_0^{\eta_n} (\sin u)^{1-\frac{1}{\beta}}\, du,\eeq
 for $n=1,\ldots, N_2$ with $N_2$ as in Lemma \ref{Mm} and where $\eta_n$ is increasing in $n$.
Then
\beq\label{vn1} v_{n+1}-v_n=\int_{\eta_{n}}^{\eta_{n+1}} (\sin u)^{1-\frac{1}{\beta}}\, du=(\sin\eta_n^*)^{1-\frac{1}{\beta}} (\eta_{n+1}-\eta_{n})\eeq
for some $\eta_n^*\in [\eta_{n},\eta_{n+1}]$. By \eqref{hn} we have, 
\beq\label{hn2}
\sin\eta_n=\frac{H_n}{s_n^{\beta} }.
\eeq
 By (3.3) in  \cite{Z2016b} and then by (3.8) in the same paper,
\begin{align}\label{phin+1n}
\eta_{n+1}-\eta_{n}=s_n^{\beta-1}+s_{n+1}^{\beta-1}=2s_n^{\beta-1}+\cO\(s_n^{\beta-1}/n\).
\end{align}
Now combining the above and recalling that $\alpha=\beta/(\beta-1)$, we rewrite \eqref{vn1} as
\begin{align}\label{vn+1}
v_{n+1}-v_n=2 H_n^{\frac{1}{\alpha}}+{\cO(N^{-1} n^{-1})}.
\end{align}
Recalling \eqref{hn}, if we use a dummy variable and sum \eqref{vn+1} from $1$ to $n$, we get
\beq\label{vnn}v_n=2nH_n^{\frac{1}{\alpha}}+\cO(N\ln n)=2nC_N^{\frac{1}{\alpha}}+\cO(\ln n/N)+\cO(s_n^{2\beta-1}N^{\alpha-1}).\eeq
In particular, for $n = N_2 = N/2 + \cO(1)$ we get $${\int_{0}^{\pi/2}(\sin u)^{1-\frac{1}{\beta}}\,du} =NC_N^{\frac{1}{\alpha}}+\cO(\ln N/N).$$
Thus
\begin{align}\label{eq:dN}
c_h=\(\int_{0}^{\pi/2}(\sin u)^{1-\frac{1}{\beta}}\,du\)^{\alpha} +\cO( \ln N/N)=I_1^{\alpha}+\cO(\ln N/N).
\end{align}

 The above implies that the set $\cup_{n=1}^{N-1}  T^n M_N\subset\cM$  is bounded by the  line $r=r'$ (or $r=r''$), $\varphi=0$, $\varphi=\pi$  and a curve described implicitly by the equation
 \beq\label{rbeta}r^{\beta}= \frac{C_N}{ \sin\varphi}(1+\cO(r^{2\beta-1}C_N^{-1}))\eeq
or equivalently, using a Taylor expansion on the $\beta$th root of the second factor on the right,
{$$r= \frac{C_N^{\frac{1}{\beta}}}{ \sin^{\frac{1}{\beta}}\varphi}+\cO\left(\frac{r^{2\beta-1}}{C_N^{1-\frac{1}{\beta}}\sin^{\frac{1}{\beta}}\varphi}\right).$$}
Using the above the calculate the tail of the return time function, we have
 \begin{align}\label{muRn}
 \mu(x\in \cM\,:\, \cR\geq N)&=\sum_{m\geq N}\sum_{k=0}^{m-N} \mu(T^k M_m)\nn\\
 &={\frac{1}{\alpha}\sum_{m\geq N}\sum_{k=0}^{m-1} \mu(T^k M_m)+\cO(N^{-\alpha})}\nn\\
 &={\frac{2}{2{\alpha|\partial Q|}}\int_{0}^{\pi}\left(\frac{C_N^{\frac{1}{\beta}}}{ \sqrt[\beta]{\sin\varphi}}+\cO\left(\frac{r^{2\beta-1}}{C_N^{1-\frac{1}{\beta}}\sin^{\frac{1}{\beta}}\varphi}\right)\right)\,  \sin\varphi\,d\varphi\nonumber+\cO(N^{-\alpha})}\\
 &=\frac{C_N^{\frac{1}{\beta}}}{\alpha|\partial Q| }\int_{0}^{\pi} (\sin\varphi)^{1-\frac{1}{\beta}}\,d\varphi+{\cO(N^{-\frac{\beta(3\beta-2)}{(\beta-1)(2\beta-1)}}+N^{-\frac{\beta}{\beta-1}})}\nonumber\\
 &=\frac{2I_1 C_N^{\frac{1}{\beta}}}{{\alpha}|\partial Q| }+\cO(N^{-\frac{\beta}{\beta-1}})\end{align}
 where we have used (\ref{rbeta}) to estimate:
\begin{align*}
\int_0^{\pi/2}\frac{r^{2\beta-1}}{C_N^{1-\frac{1}{\beta}}\sin^{\frac{1}{\beta}}\varphi}\, \sin\varphi \, d\varphi&=\int_{0}^{\pi/2}\frac{ r^{2\beta-1}\sin^{1-\frac{1}{\beta}}\varphi}
{C_N^{1-\frac{1}{\beta}}}\,d\varphi\\
&=\int_{0}^{\pi/2}  r^{3\beta-2}\, d\varphi=\cO(N^{-\frac{\beta(3\beta-2)}{(\beta-1)(2\beta-1)}}).\end{align*}

Now combining \eqref{muRn} with \eqref{eq:dN} and the definition of $C_N$, we get
$$\lim_{N\to\infty}N^{\frac{1}{\beta-1}}\mu(x\in \cM\,:\, \cR\geq N) ={\frac{2I_1^{\alpha}}{\alpha|\partial Q| }.}$$
Using the fact that
$$\mu(x\in \cM\,:\, \cR\geq N)=\sum_{m\geq N}\mu( x\in M\,:\, \cR\geq m)+\cO(N^{-\alpha})$$
and the fact that $\frac{1}{\beta-1}=\alpha-1$ we obtain
\begin{align}\nn\lim_{N\to\infty}N^{\alpha}\tmu(x\in M\,:\, \cR\geq N)&=\frac{2I_1^{\alpha}}{\alpha(\beta-1)\mu(M)|\partial Q|}\nn\\
&=\frac{2I_1^{\alpha}}{\beta\mu(M)|\partial Q|}.
 \end{align}
This allow us to calculate
\begin{align*}
\tilde\sigma_\cR^\alpha=\lim_{N\to\infty} N\tmu\left(\cR > N^{\frac{1}{\alpha}}\right)
&=\frac{2I_1^{\alpha}}{\beta\mu(M)|\partial Q|}.
\end{align*}
which completes the proof of Lemma \ref{Mm2}.
\end{proof}

In order to extend the calculation of the scale parameter to general $\tf$, we need to prove the error bound in Lemma \ref{E bound}.

\begin{proof}[Proof of Lemma \ref{E bound}]
Consider any $\gamma$-H\"{o}lder continuous function $f$ on $\cM$.
Similar to the proof of Lemma \ref{Mm2}, we write
\beq\nn v=\Psi(\eta):=\int_0^{\eta} (\sin u)^{\frac1{\alpha}}\, du,\eeq
where we have set the variable $v$ equal to the function $\Psi(\eta)$ since want to use $$\eta=\Psi^{-1}(v)\in(0,\pi/2).$$
In particular,
\beq\label{Psi-1}
\frac{d\Psi^{-1}(v)}{d v}=(\sin\eta)^{-\frac{1}{\alpha}}.
\eeq
By \eqref{hn2},
\begin{align*}
\frac{d\Psi^{-1}(v_n)}{d v}&=(\sin\eta_n)^{-\frac{1}{\alpha}}\\
&= \left(\frac{H_n}{s_n^{\beta} }+\cO(s_n^{\beta-1}) \right)^{-\frac{1}{\alpha}}\\
&= \frac{s_n^{\beta-1}}{H_n^{\frac{1}{\alpha}}}\left(1+\cO(s_n^{2\beta-1}) \right)^{-\frac{1}{\alpha}}= \frac{s_n^{\beta-1}}{H_n^{\frac{1}{\alpha}}}\left(1+\cO(s_n^{2\beta-1}) \right)\\
&=\frac{s_n^{\beta-1}}{H_n^{\frac{1}{\alpha}}}+\cO(s_n^{3\beta-2}H_n^{-\frac{1}{\alpha}})\end{align*}
thus using (\ref{vn+1}), we know that
\beq\label{angle bound}
|\Psi^{-1}(v_{n+1})-\Psi^{-1}(v_n)|\leq 3 s_n^{\beta-1}
\eeq
Combining \eqref{eq:dN}  with (\ref{vnn}) gives
$$v_n=\frac{{2nI_1}}{N}+\cO(\ln n/N).$$
By the Mean Value Theorem, there exists $\eta_*$ lying between $\eta_{n}=\Psi^{-1}(v_n)$ and $\eta_{n+1}=\Psi^{-1}(v_{n+1})$ such that
$$\Psi^{-1}(v_n)-\Psi^{-1}(v_{n+1})= (\sin\eta_*)^{-\frac{1}{\alpha}} (v_n-v_{n+1}) $$
thus
\begin{align}\label{Psins}
\eta_n&=\Psi^{-1}\(2n I_1/N+\cO(\ln n/N)\)\nonumber\\
&=\Psi^{-1}(2I_1n/N)+ (\sin\eta_*)^{-\frac{1}{\alpha}}\cO(\ln n/N)\nn\\
&=\Psi^{-1}(2I_1n/N)+\cO(s_n^{\beta-1}\ln n).
\end{align}

We now estimate the sum $\cS_{N_2}:=
\sum_{n=1}^{N_2} f(r_n,\varphi_n)$.
Note that collisions with the curves $\Gamma_1$ and $\Gamma_2$ alternate. Thus it is convenient to introduce: $$\bar f(\varphi)=(f(r',\varphi)+f(r'',\pi-\varphi))/2$$ where $r',r''$ are as in \eqref{r'}.

Also,  recall that the function $f$ is H\"{o}lder continuous with exponent
 $\gamma $ in the variables $(r,\varphi)$.
Thus we have
\begin{align*}
f(r,\Psi^{-1}(v_{n+1}))-f(r,\Psi^{-1}(v_n))&\leq \|f\|_{\gamma}|\Psi^{-1}(v_{n+1})-\Psi^{-1}(v_n)|^{\gamma } \leq C\|f\|_{\gamma}s_n^{(\beta-1)\gamma}.
\end{align*}
Moreover, if $r_n$ is based on the same boundary as $r'$, noting that $s_n=|r_n-r'|+\cO(|r_n-r'|^2)$, then using (\ref{phin+1n}), we get
\beq\label{r bound}
|f(r_n,\varphi_n)-f(r',\varphi_n)|\leq C\|f\|_{\gamma}s_n^{\gamma(\beta-1)}.
\eeq
Recall that $$I_f=\frac{1}{2}\int_{0}^{\pi} \bar f(\varphi) (\sin\varphi)^{1-\frac{1}{\beta}}\,d\varphi.$$
Using Lemma \ref{Mm}, \eqref{angle bound}, \eqref{Psins}, and \eqref{r bound}  gives
\begin{align*}
\cS_{N_2} f
&=\sum_{n=1}^{N_2} (f(r_n,\varphi_n))\\
&=\sum_{n=1}^{N_2}\frac{f(r',\varphi_n)+f(r'',{\pi-\varphi_n})}{2}+ {\cO\(\sum_{n=1}^{N_2}s_n^{\gamma}\)+\cO\(\sum_{n=1}^{N_2}s_n^{(\beta-1)\gamma}(\ln n)^{\gamma}\)}\\
&=\sum_{n=1}^{N_2}\frac{f(r',\varphi_n)+{f(r'',\pi-\varphi_n)}}{2}+ {\cO\(\sum_{n=1}^{N_2}s_n^{\gamma}\)}\\
&=\sum_{n=1}^{N_2}  \bar f(\varphi_n) + \cO\(\sum_{n=1}^{N_2}(nN^{\alpha})^{-\frac{\gamma}{(\alpha+1)(\beta-1)}}\)\\
&={\sum_{n=1}^{N_2}\(\bar f(\Psi^{-1}(2I_1 n/N))+\cO(N^{\frac{2\gamma}{\alpha+1}} n^{-(1+\frac{2}{\alpha+1})\gamma})\)}+\cO(N^{1- \frac{\gamma}{\beta-1}})\\
&={\frac{N}{2I_1}}\int_0^{I_1}\bar f(\Psi^{-1}(v))\,dv +\cO(N^{1- \gamma})+\cO(N^{1- \frac{\gamma}{\beta-1}})\\
&=\frac{N}{{2I_1}}{\int_0^{\pi/2}}\bar f(\varphi)\cdot (\sin\varphi)^{1-\frac{1}{\beta}} \,d\varphi+\cO(N^{1- \frac{\gamma}{\beta-1}})\\
&{=\frac{N I_f}{{2I_1}}}+\cO(N^{1- \frac{\gamma}{\beta-1}}).\end{align*}
By time reversibility, the trajectory going out of the cusp during
the period $N_2\leq n\leq N$ has similar properties. Thus
$$\sum_{n=N_2+1}^N f(r_n, \varphi_n)=\frac{N-N_2}{2I_1}{\int_{\pi/2}^{\pi}}\bar f(\varphi)\cdot (\sin\varphi)^{1-\frac{1}{\beta}} \,d\varphi+\cO(N^{1- \frac{\gamma}{\beta-1}})={\frac{N I_f}{{2I_1}}}+\cO(N^{1- \frac{\gamma}{\beta-1}}).$$
Therefore we can get  the following estimation on the sum for $x\in M_N$, \beq\nn\tilde f(x) =\(\frac{I_f}{I_1}\)(\cR(x)-\tmu(\cR))+E(x)\eeq
where for some $C>0$
\beq\nn\label{errorbound}
|E(x)|\leq C\cR(x)^{1- \frac{\gamma}{\beta-1}}.\eeq
\end{proof}
\section{Proof of Proposition \ref{distortionbound}}\label{sec:7}

 We will show that curves in $\cW^u_H$  have uniformly bounded curvature and uniform distortion bounds.
\subsection{Bounded curvature.}   For any $x\in M$, and any unit tangent vector $dx=(dr,d\varphi)\in \cT_{x}M$.
\begin{lemma}[Curvature bounds]\label{curvbd} Fix $C_0>0$. There exists a constant $C\geq C_0$, such that for any twice continuously differentiable unstable curve $W$,  with curvature bounded  by $C_0$, we have that every curve $W'$, such that  $W'\subset  F W$, is also twice differentiable with its curvature bounded by $C$:
\beq\label{slopebd}|d^2\varphi/dr^2|\leq C. \eeq
\end{lemma}
 \begin{proof}
 Fix $N_0>1$. We consider two cases: the first being $M'=\{M_N,\,\, N\leq N_0\}$, and the second being $M\setminus M'$.

 For any $x\in M'$, we can find bounds $\cK_{\max}>\cK_{\min}>0$ and $\tau_{\max}>\tau_{\min}>0$, such that the curvature of the boundary at the base of $x$ satisfies $\cK(x)\in [\cK_{\min},\cK_{\max}]$, and the length of the free path $\tau(x)\in  [\tau_{\min},\tau_{\max}]$. This implies that restricted to $M'$, the map $F$ is a dispersing billiard map. Thus by \cite{chernov2006chaotic}~Proposition 4.29, (\ref{slopebd}) holds for $x$ and its forward iterations $T^k x$, $k=1,\cdots, N$. This also implies that every smooth curve $W'\subset F W(x)$ has the same property.

  Next we consider $x\in M\setminus M'$.  First note that for any $x\in M\setminus M'$, the free paths corresponding to the first and last collisions have lengths uniformly bounded from below:
  \beq\label{taumin}
  \tau(x)\geq \tau_{\min}.
  \eeq
  For long series of collisions in the corner, the unstable manifolds $T^k W(x)$, $k=1,\cdots, N$, can be approximated by the singular curves which form the boundaries of $H_N$, according to Lemma \ref{Mm}.
  Since we are interested in the last collision in the corner series, we estimate the slope of the tangent vectors of the boundary of $H_N\cap T^N M_N$  on the collision space of $\Gamma_1$ (or $\Gamma_2$ by symmetry).

  Let $y=(r,\varphi)\in T^NW$ be a point in $T^{N}M_N$ that lies on the long boundary of $H_N$.  Moreover, set $y_0=(r_0, \varphi_0)\in T^{N-1}M_N$ such that $y=Ty_0$ and also set $y_1=Ty=(r_1,
  \varphi_1)\in FW=T^{N_1}W$. In order to calculate the slope unstable vector at $y_1$ for $FW$, we will first calculate that of $y_0$ at $T^{N-1}W$, which also gives an estimation for   the slope of the tangent vector at $y\in T^{N}W$.

  Using  item (8) of Lemma \ref{Mm},  one can check that the slope of the tangent vector to the curve $T^{N-1}W$ at $y_0=(r_0,\varphi_0)$ satisfies \beq\label{slopeHN}\cV_0:=d\varphi_0/dr_0\sim \cot \varphi_0/r_0.\eeq

  According to the differential formula for billiards (2.26) in  \cite{chernov2006chaotic},  if we denote $\cV=d\varphi/dr$ as the slope of the tangent vector $d y=D T  dy_0=(dr, d\varphi)$, then it satisfies
\begin{align}\label{slope}
\cV &=\frac{\tau \cK\cK_0+\cK_0\cos\varphi+\cK\cos\varphi_0+(\tau_0\cK+\cos\varphi)\cV_0}{\tau_0\cK_0+\cos\varphi_0+\tau_0 \cV_0}\\
&=\nn \cK+\frac{\cos\varphi(\cK_0+\cV_0)}{\tau_0(\cK_0+ \cV_0)+\cos\varphi_0}\sim \cK\sim N^{-\frac{\beta(\beta-2)}{(\beta-1)(2\beta-1)}},
\end{align}
where $\tau_0$ is the length of the free path between $y_0$ and $y$,  $\cK=\cK(y)$, and $\cK_0=\cK(y_0)$. Here, we used the estimates in Lemma \ref{Mm} for the last step. Similarly, one can check that
\begin{align*}
d\cV/dr&\sim N^{-\frac{\beta(\beta-3)}{(\beta-1)(2\beta-1)}}.
\end{align*}
This implies that   for an unstable curve $W\subset M\setminus M'$,  the slope of the unstable curve $T^{N} W$ is approximately $\cV\sim  N^{-\frac{\beta(\beta-2)}{(\beta-1)(2\beta-1)}}$, and its curvature $\sim N^{-\frac{\beta(\beta-3)}{(\beta-1)(2\beta-1)}}.$
We use the differential formula (2.26) in \cite{chernov2006chaotic} again, to get the slope for the curve $FW=T^{N+1}W$, which for $y_1=(r_1,\varphi_1)=Ty\in FW=T^{N+1}W$, $y=(r,\varphi)\in T^N W$, can be approximated as
$$d\varphi_1/dr_1=\frac{\tau \cK\cK_1+\cK\cos\varphi_1+\cK_1\cos\varphi+(\tau\cK_1+\cos\varphi_1)\cV}{\tau\cK+\cos\varphi+\tau \cV}\sim \cK_1+\frac{\cos\varphi_1}{\tau}$$
where $\tau=\tau(y)$, and $\cK=\cK(y)$, $\cK_1=\cK(y_1)$. Note that $F x$ is a unit vector on $\Gamma_3$ (the boundary portion of the billiard table opposing the cusp), thus $\cK(F x)\leq \cK_{\max}$, and by (\ref{taumin}), the length of the free path for the last collision satisfies $\tau(T^N x)>\tau_{\min}$. This implies that $|d\varphi_1/dr_1|\leq \cK_{\max}+\tau_{\min}^{-1}$.
Similarly, one can check that $d^2\varphi_1/dr_1^2$ is also uniformly bounded on $FW$.
    \end{proof}

\noindent Remark:  From now on whenever we mention an unstable curve $W$, we assume that  $W$ is contained in $FM_N$, for some $N\geq 1$, and that every  iteration $T^{-j} W$, $j=1,\ldots,N$ is homogeneous. In particular, this requires that $T^{-1}W$ is contained in a single homogeneity strip $\bH_k$.    For any $x=(s_x,\varphi_x)$ and $y=(s_y,\varphi_y)\in T^{-1}W\subset \bH_k$, we then have that $\cos\varphi_x\sim \cos\varphi_y \sim 1/k^2$.  If $T^{-1}W$  stretches fully in $ \bH_k$, then the stretch in the $\varphi$-dimension of $T^{-1}W$ is approximately $k^{-3}$. By (\ref{slope}) we know that the slope of a tangent vector to $T^{-1}W\subset \bH_k$ is approximately $N^{-\frac{\beta(\beta-2)}{(2\beta-1)(\beta-1)}}$. Thus
\beq\label{|W|}|T^{-1}W|\sim N^{\frac{\beta(\beta-2)}{(2\beta-1)(\beta-1)}}\cos\varphi_1^{3/2}.\eeq 

Now using  \cite{Z2016b}~Lemma 18 we obtain the expansion factor
\beq\label{lambda11}\frac{\|dx\|}{\|dx_1\|}\sim \frac{\|dx\|_p}{\|dx_1\|_p}\cdot \frac{\cos\varphi_1}{\cos\varphi}\sim  \cK(x_1)\sim  N^{-\frac{\beta(\beta-2)}{(2\beta-1)(\beta-1)}}.\eeq
Thus by the bounded curvature property, for any $x\in W$, with $T^{-1}x=(r_1,\varphi_1)$,
\beq\label{Wcos}|W|\sim  |T^{-1}W| \max_{x\in W} \frac{\|dx\|}{\|dx_1\|}\sim \cos\varphi_1^{3/2}.\eeq
Moreover, note the following observation regarding an unstable curve $W\subset FM_N$, such that $F^{-1}W\subset M_N$ is an unstable curve. Although the slope of the tangent line for $T^{-1}W$ is given by (\ref{slope}); its backward images under $T^{-n}$, for $n=2,\cdots, N-1$, satisfy \beq\label{slope1}
|d\varphi_n/dr_n|\sim \frac{\cot\varphi_n}{r_n}\eeq according to (\ref{slopeHN}) by using the estimates for the boundary of $H_N$ given in item (8) of Lemma \ref{Mm}.

  \subsection{Distortion bounds of the Jacobian under
	$F$.}\label{sec:distortion2}
  Denote by $J_W^n(x)$ the Jacobian of $F^n$ along $W$ at $x\in W$. In the following proof we use $C$, depending only on the billiard table,  to represent a generic positive constant which may change from line to line, and sometimes within the same line.
\begin{lemma}[Distortion bounds]\label{distorbd} Let $W\subset FM_m\subset \bH_k$ be an unstable curve, for some $k$ large and $m\geq 1$. Then
\beq\label{distortion2}|\ln J_WF^{-1}(x)-\ln J_WF^{-1}(y)|\leq C_b d_W(x,y)^{a},\eeq for any $a\in (0,1/3)$, where   $C_b>0$ is a constant depending only on the billiard table. \end{lemma}
\begin{proof}

Again we fix $N_0>1$, and consider the two cases $M'=\{M_N,\,\, N\leq N_0\}$ and $M\setminus M'$.

 When restricted to $M'$, as argued before, the map $F$ is a dispersing billiard map. Thus by \cite{chernov2006chaotic}~Lemma 5.27, the distortion bounds hold, for $x\in M_N$, with $N\leq N_0$, and its forward iterations $T^k x$, with $k=1,\cdots, N$, with Holder exponent $1/3$. This also implies that (\ref{distortion2}) holds for $F$.

Let $x_{n}= T ^{-n} x=(r_{n}, \varphi_{n})\in W_{n}=T^{-n}W$, for $x$ belonging to some unstable curve $W\subset FM_N$ with  $N>N_0$, and with $n=0,\cdots, N+1$.  Let $\cK_{n}=\cK(x_{n})$ and $\tau_{n}=\tau(x_{n}).$  For a tangent vector $dx=(dr, d\varphi)\in \cT_xM$, let its  $p$-norm be given by $\|dx\|_p=|dr|\cos\varphi$, and its Euclidean norm be given by $\|dx\|=\sqrt{|dr|^2+|d\varphi|^2}$.

Using formula (\ref{slope1}), we set
	\beq\label{eq:nhat2}
	\hat N=CN^{\frac{\beta(\beta-2)}{(\beta+1)(\beta-1)}}
	\eeq
	 so that $\cot\varphi_{\hat N}\sim r_{\hat N}$, and in particular, for $n=2,\cdots, \hat N$ or $n=N-\hat N, \cdots, N$, $|d\varphi_n/d r_n|\leq 1$, while for $n=\hat N, \cdots, N-\hat N$, we have $|d\varphi_n/d r_n|\geq 1$.

We have\begin{align*}	 J_{W_{n}}T^{-1}(x_n)&=\frac{1}{J_{T^{-1}W_{n}}(x_{n+1})}=\frac{\|dx_{n+1}\|}{\|dx_{n}\|}\\
&=\frac{\|dx_{n+1}\|_p}{\|dx_{n}\|_p}\cdot \frac{\cos\varphi_{n}}{\cos\varphi_{n+1}}\cdot\sqrt{\frac{1+(d\varphi_{n+1}/dr_{n+1})^2}{1+(d\varphi_{n}/dr_{n})^2}}\\
&=\frac{1}{1+\tau_{n+1}(d\varphi_{n+1}/dr_{n+1}+\cK_{n+1})/\cos\varphi_{n+1}}\frac{\cos\varphi_{n}}{\cos\varphi_{n+1}}\sqrt{\frac{1+(d\varphi_{n+1}/dr_{n+1})^2}{1+(d\varphi_{n}/dr_{n})^2}}\\
&=\frac{\cos\varphi_n}{\cos\varphi_{n+1}+\tau_{n+1}(\cK_{n+1}+d\varphi_{n+1}/dr_{n+1})}\sqrt{\frac{1+(d\varphi_{n+1}/dr_{n+1})^2}{1+(d\varphi_n/dr_n)^2}}.\end{align*} Therefore
\begin{align}
   \ln  J_{W_{n}}T^{-1}(x_n) &=
   \ln Z(x_n)
  +\ln \cos\varphi_n -
	\ln P(x_{n})	\label{lncJnBB}
\end{align}
where $$P(x_n)= \cos\varphi_{n+1}+\tau_{n+1}(\cK_{n+1}+d\varphi_{n+1}/dr_{n+1})\,\,\,\,\,\,\text{ and }\,\,\,\,\,\,\,\, Z(x_n)=\sqrt{\frac{1+(d\varphi_{n+1}/dr_{n+1})^2}{1+(d\varphi_n/dr_n)^2}}.$$
Differentiate (\ref{lncJnBB}) with respect to $x_n$ to get
\begin{align*}
\frac{d}{dx_n}\ln J_{W_{n}}T^{-1}(x_n)&=\frac{d}{dx_n}\ln Z(x_n)+\tan\varphi_n  \frac{d\varphi_n }{dx_n}-\frac{d}{dx_n}\ln P(x_n).
\end{align*}
Using Lemma \ref{Mm}, one can check that for $n=1,\cdots, N$, $$P(x_n)\sim \cos\varphi_{n+1},\,\,\text{and }\,\,\,P(x_0)\sim \cK_1.$$ Moreover, note that
\begin{align*}
\frac{d}{d x}=\frac{d x_n}{dx}\cdot \frac{d}{dx_n}=J_{W_n}T^{-n}(x)\frac{d}{dx_n}.
\end{align*}
Also,\begin{align}\label{expansion}\frac{\|dx_n\|}{\|dx\|}&=\frac{\|dx_{n}\|_p}{\|dx\|_p}\cdot \frac{\cos\varphi}{\cos\varphi_{n}}\cdot\sqrt{\frac{1+(d\varphi_{n}/dr_{n})^2}{1+(d\varphi/dr)^2}}\nonumber\\
&\leq C\cos\varphi_1\cdot n^{-\frac{\beta}{2\beta-1}}N^{\frac{\beta(\beta-2)}{(2\beta-1)(\beta-1)}}\cdot \frac{\cos\varphi}{\cos\varphi_{n}}\cdot\sqrt{\frac{1+(d\varphi_{n}/dr_{n})^2}{1+(d\varphi/dr)^2}},\end{align}
according to \cite{Z2016b}~Proposition 4, for $n=1,\cdots, N$. Here we also note that if $x\in FM_N$, for $N>N_0$, with $x=(r,\varphi)$, then $\cos\varphi>C_{N_0}$. In other words, $\cos\varphi$ is uniformly bounded away from $0$ by some constant $C_{N_0}>0$.

By (\ref{slope1}) and the definition of $\hat N$,
 for $n=1,\cdots, \hat N$ and for any smooth function $f(x_n)$,
 $$\frac{1}{\sqrt{2}}\left|\frac{d f(x_n)}{ d r_n}\right|\leq \left|\frac{d f(x_n)}{d x_n}\right|= \left|\frac{d f(x_n)}{ d r_n}\right|\cdot \frac{1}{\sqrt{1+(d\varphi_n/dr_n)^2}}\leq \left|\frac{d f(x_n)}{ d r_n}\right|.$$
 Thus it is enough to consider $d f(x_n)/dr_n$. Using Lemma \ref{Mm} and Equation (3.32) in \cite{Z2016b}~Proposition 2, as well as the estimates for $\tau_n$ in the last paragraph of the proof there, we know that \beq\label{tauKcos}\cK_{n+1}\tau_{n+1}/\cos\varphi_{n+1} \sim n^{-2}, \,\,\,\,\,\,\,\tau_{n+1}\sim n^{-\frac{2\beta}{2\beta-1}}N^{-\frac{\beta}{(2\beta-1)(\beta-1)}}.\eeq Using (\ref{slope1}), this implies that
$$\left|\frac{d}{dx_n}\ln Z(x_n)\right|\leq C/r_n\cdot |d\varphi_n/dr_n|\leq C \cos\varphi_n/r_n^2$$
and
\begin{align*}
\left|\frac{d}{dx_n}\ln P(x_n)\right|&\leq \frac{C}{\cos\varphi_{n+1}}\left(C\cos\varphi_{n+1}/r_{n+1}+C\tau_{n+1} r_{n+1}^{\beta-3}+C\cos\varphi_{n+1} \tau_{n+1} /r_{n+1}^2\right)\\
&\leq \frac{C}{r_{n+1}}\cdot\left(C+Cn^{-2}+C/n\right)\leq C/r_{n+1}.
\end{align*}

Now we consider
$$\tan\varphi_n  \left|\frac{d\varphi_n }{dx_n}\right|\leq C\frac{1}{\cos\varphi_n}\left|\frac{d\varphi_n}{dr_n}\right|.$$
For $n=1,\cdots, \hat N$, using (\ref{slope1}), we get
$$\tan\varphi_n  \left|\frac{d\varphi_n }{dx_n}\right|\leq C/ r_n,$$
while for $n=0$,
we get
$$\tan\varphi  \left|\frac{d\varphi}{dx}\right|\leq C.$$
Moreover, by (\ref{slope}) and (\ref{lambda11}), together with Lemma \ref{Mm}, for $n=0$, we have
$$\left|\frac{d}{dx}\ln Z(x)\right|\leq  C\cK_1 r_1^{\beta-3} \cdot \frac{\|dx_1\|}{\|dx\|}\leq C r_1^{\beta-3},$$
$$ \left|\frac{d}{dx}\ln P(x)\right|\leq\frac{C\tau_1\cK_1+C\tau_1r_1^{\beta-3}}{\tau_1\cK_1 }\cdot \frac{\|dx_1\|}{\|dx\|}\leq C r_1^{\beta-3}.$$
Similarly, for $n=N$, we can check that
$$\frac{d}{dx_N}\ln J_{W_N}T^{-1}(x_N)\leq \frac{C}{r_{N}}.$$
By \cite{Z2016b}~Proposition 4
\begin{align}\label{expansionM}\frac{\|dx_N\|}{\|dx\|}&=\frac{\|dx_{N}\|_p}{\|dx\|_p}\cdot \frac{\cos\varphi}{\cos\varphi_{N}}\cdot\sqrt{\frac{1+(d\varphi_{N}/dr_{N})^2}{1+(d\varphi/dr)^2}}\nonumber\\
&\leq C\cos\varphi_1\cos\varphi_N\cdot N^{-\frac{\beta}{2\beta-1}}N^{-\frac{\beta-1}{2\beta-1}}r_N^{1-\beta} r_1^{2-\beta}\cdot \frac{\cos\varphi}{\cos\varphi_{N}}\cdot\sqrt{\frac{1+(d\varphi_{N}/dr_{N})^2}{1+(d\varphi/dr)^2}}\nonumber\\
&\leq CN^{-1}r_N^{3-2\beta}.\end{align}
Thus
$$\frac{d}{dx}\ln J_{W_N}T^{-1}(x_N)\leq \frac{C}{r_{N}} \cdot \frac{\|dx_N\|}{\|dx\|}\leq CN^{-1}r_N^{2-2\beta}\leq C N^{\frac{1}{2\beta-1}}.$$
Combining the above facts, we have that \beq\label{dost1}\frac{d}{dx}\ln J_{W}T^{-1}(x)\leq C r_1^{\beta-3}\eeq and for $n=1,\cdots, \hat N$, we have
\beq\label{dist2}\frac{d}{dx_n}\ln J^{-1}_{W_n}(x_n)\leq C/ r_n.\eeq
By symmetry, one can also show that for $n=N-\hat N, \cdots, N$, (\ref{dist2}) also holds.

For $n=2,\cdots,\hat N$, using (\ref{expansion}),
\begin{align}\label{beforehatN}&\frac{d}{dx}\ln J_{W_n}T^{-1}(x_n)\leq \frac{C}{r_{n}} \cdot \frac{\|dx_n\|}{\|dx\|}\nonumber\\
&\leq \frac{C}{r_{n}}\cos\varphi_1\cdot n^{-\frac{\beta}{2\beta-1}}N^{\frac{\beta(\beta-2)}{(2\beta-1)(\beta-1)}}\cdot \frac{\cos\varphi}{\cos\varphi_{n}}\\&\leq
C n^{-\frac{\beta}{2\beta-1}}N^{\frac{\beta(\beta-2)}{(2\beta-1)(\beta-1)}}\cdot \frac{\cos\varphi_1}{r_n\cos\varphi_n}\nonumber\\
&\leq C n^{-\frac{\beta}{2\beta-1}}N^{\frac{\beta(\beta-2)}{(2\beta-1)(\beta-1)}} \cdot n^{\frac{1}{2\beta-1}}N^{\frac{\beta}{(2\beta-1)(\beta-1)}}\cdot N^{\frac{\beta}{2\beta-1}} n^{-\frac{\beta}{2\beta-1}}\cdot\cos\varphi_1\nonumber\\
&=C n^{-1}N^{\frac{2\beta}{2\beta-1}}\cdot\cos\varphi_1\nonumber.\end{align}
For $n=1$, by (\ref{expansion}),
\begin{align*}\frac{d}{dx}\ln J_{W_1}T^{-1}(x_1)&\leq \frac{C}{r_{1}}\cos\varphi_1\cdot N^{\frac{\beta(\beta-2)}{(2\beta-1)(\beta-1)}}\cdot
\frac{\cos\varphi}{\cos\varphi_{1}}\leq
C r_1^{\beta-3}.\end{align*}
Thus
\beq\label{dist3}\sum_{n=1}^{\hat N} \frac{d}{dx}\ln J_{W_n}T^{-1}(x_n)\leq C \ln \hat N\cdot N^{\frac{2\beta}{2\beta-1}}\cdot\cos\varphi_1\leq C\ln N\cdot N^{\frac{\beta}{2\beta-1}}.\eeq

For $n=\hat N,\ldots,N-\hat N$, we know that $d\varphi_n/dr_n\geq 1$ so we use another formula for the Jacobian. Now, denoting $\hat\cV_n=dr_n/d\varphi_n$,
\begin{align}\label{2JTW}	& J_{W_{n}}T^{-1}(x_n)=\frac{1}{J_{T^{-1}W_{n}}(x_{n+1})}=\frac{\|dx_{n+1}\|}{\|dx_{n}\|}=\frac{|d\varphi_{n+1}|}{|d\varphi_{n}|}\cdot\sqrt{\frac{1+\hat\cV_{n+1}^2}{1+\hat\cV_n^2}}\\
&=\left(\tau_n\cK_{n+1}+\cos\varphi_{n+1}+(\tau_n\cK_n\cK_{n+1}+\cK_n\cos\varphi_{n+1}+\cK_{n+1}\cos\varphi_{n+1})\hat\cV_n\right) \sqrt{\frac{1+\hat\cV_{n+1}^2}{1+\hat\cV_n^2}}.\nonumber\end{align} Therefore  for $n=\hat N, \cdots, N-\hat N$,
\begin{align}
   \ln  J_{W_{n}}T^{-1}(x_n) &=
   \ln \tZ(x_n)
 +
	\ln \tP(x_{n+1})	\label{lncJnBB1}
\end{align}
where $$\tP(x_n)= \tau_n\cK_{n+1}+\cos\varphi_{n+1}+(\tau_n\cK_n\cK_{n+1}+\cK_n\cos\varphi_{n+1}+\cK_{n+1}\cos\varphi_{n+1})\hat\cV_n$$ and $\tZ(x_n)= \sqrt{\frac{1+\hat\cV_{n+1}^2}{1+\hat\cV_n^2}}.$

By (\ref{slope1}) and the definition of $\hat N$,
 for $n=\hat N,\cdots, N-\hat N$, and for any smooth function $f(x_n)$,
 $$\frac{1}{\sqrt{2}}\left|\frac{d f(x_n)}{ d \varphi_n}\right|\leq \left|\frac{d f(x_n)}{d x_n}\right|= \left|\frac{d f(x_n)}{ d \varphi_n}\right|\cdot \frac{1}{\sqrt{1+(dr_n/d\varphi_n)^2}}\leq \left|\frac{d f(x_n)}{ d \varphi_n}\right|.$$
 Thus it is enough to consider $d f(x_n)/d\varphi_n$ for $n=\hat N,\cdots, N-\hat N$.

Differentiate (\ref{lncJnBB1}) with respect to $x_n$ and get
\begin{align*}
\frac{d}{dx}\ln J_{W_{n}}T^{-1}(x_n)&=\frac{d}{dx}\ln \tZ(x_n)+\frac{d}{dx}\ln \tP(x_n).
\end{align*}
By (\ref{slope1}), one can check that for $n=1,\cdots, \hat N$,
$$\left|\frac{d}{dx_n}\ln \tZ(x_n)\right|\leq C,\,\,\,\,\,\,\left|\frac{d}{dx_n}\ln \tP(x_n)\right|\leq C$$
for some constant $C>0$.
Summing over $n=\hat N,\cdots, N-\hat N$, we have
\begin{align*}
&\sum_{n=\hat N}^{N-\hat N}\frac{d}{dx}\ln J_{W_n}T^{-1}(x_n)\leq C \sum_{n=\hat N}^{N-\hat N}\frac{\|dx_{n}\|}{\|dx\|}\\
&\leq  C \sum_{n=\hat N}^{N-\hat N}\cos\varphi_1\cdot n^{-\frac{\beta}{2\beta-1}}N^{\frac{\beta(\beta-2)}{(2\beta-1)(\beta-1)}}\cdot \frac{\cos\varphi}{\cos\varphi_{n}}\\
&\leq C \sum_{n=\hat N}^{N-\hat N} n^{-\frac{2\beta}{2\beta-1}}N^{\frac{\beta(\beta-2)}{(2\beta-1)(\beta-1)}}\\
&\leq C \hat N^{-\frac{1}{2\beta-1}}N^{\frac{\beta(\beta-2)}{(2\beta-1)(\beta-1)}}\leq C_5 N^{\frac{\beta^2(\beta-2)}{(2\beta-1)(\beta-1)(\beta+1)}},
\end{align*}
where we used (\ref{beforehatN}) for the estimation of $\frac{\|dx_{\hat N}\|}{\|dx\|}$ in the last step.

If we denote $d_W(x,y)$ as the distance between $x$ and $y$ measured along $W$, then by the bounded curvature of $W$ and (\ref{Wcos}), we have $$d_W(x,y)<C|W|\sim \cos\varphi_1^{3/2}.$$ This implies
\begin{align*}
	&\bigl| \ln J_{W}T^{-1}(y)
  - \ln J_{W}T^{-1}(x) \bigr| \leq d_W(x,y) \sum_{n=0}^N\max_{x\in W}\Bigl|\frac{d}{dx}
  \ln J_{W_k}T^{-1}(x_k) \Bigr|\\&
=d_W(x,y) \left(\sum_{n=1}^{N}\max_{x\in W}\Bigl|\frac{d}{dx}
  \ln J_{W_k}T^{-1}(x_k) \Bigr| +\max_{x\in W}\Bigl|\frac{d}{dx}
  \ln J_{W}T^{-1}(x) \Bigr| \right)\\
  &\leq d_W(x,y)^{a}\left(C\ln N\cdot N^{\frac{\beta}{2\beta-1}}  +C N^{\frac{\beta^2(\beta-2)}{(2\beta-1)(\beta-1)(\beta+1)}}+r_1^{\beta-3}\right)|\cos\varphi_1|^{3(1-a)/2}\\
  &\leq
d_W(x,y)^a\left(C\ln N\cdot N^{\frac{\beta}{2\beta-1}}  +C N^{\frac{\beta^2(\beta-2)}{(2\beta-1)(\beta-1)(\beta+1)}} +C N^{\frac{\beta(3-\beta)}{(2\beta-1)(\beta-1)}}\right) \cos\varphi_1^{3(1-a)/2}\\
   &\leq C
d_W(x,y)^{a}\ln N\cdot N^{\frac{\beta}{2\beta-1}}\cos\varphi_1^{3(1-a)/2}\end{align*}
where we used (\ref{Wcos}) in the  last step. Note that,  in the second last inequality, the first term dominates. Now using the fact that $\cos\varphi_1=\cO(N^{-\frac{\beta}{2\beta-1}})$, one can check that for any $a\in (0,1/3)$,
$$\bigl| \ln J_{W}F^{-1}(y)
  - \ln J_{W}F^{-1}(x) \bigr| \leq C_b d_W(x,y)^a,$$ for some constant $C_b>0$.
\end{proof}
The above lemma exhibits the following. Regarding the expansion factor $\Lambda(x)$ (in the Euclidean metric), the function $\ln \Lambda(x)$ is less regular on points which enter longer series of corner collisions, i.e, at the points where the H\"{o}lder exponent is smaller than $1/3$. For points where the trajectory has a bounded number of collisions in the corner series, the  function $\ln \Lambda(x)$ is still H\"{o}lder continuous with exponent $1/3$, which is similar to  the situation for dispersing  billiards, see \cite{chernov2006chaotic}.

\section*{Acknowledgements}

The research of H. Zhang was supported in
part by NSF grant DMS-1151762, and also in part by a grant from the Simons
Foundation (337646, HZ). The research of P. Jung was supported in part by NSA grant H98230-14-1-0144 and NRF grant N01170220. We would like to thank Dmitry Dolgopyat for posing the questions and also suggesting the main results discussed in this paper, i.e., the emergence of stable laws in billiard systems exhibiting slow decay of correlations. H. Zhang also thanks him for many invaluable discussions and suggestions.

\bibliographystyle{alpha}

\end{document}